\documentclass{article}[12pt]
\usepackage{amsthm}
\usepackage{amsmath,amssymb}
\usepackage{float}
\usepackage{geometry}
\usepackage{indentfirst}
\usepackage{array,multirow}
\usepackage[table,xcdraw]{xcolor}
\usepackage{bm}
\usepackage{tabularx,booktabs}
\usepackage{lscape}
\usepackage{enumitem}
\usepackage{chngcntr}
\usepackage[table]{xcolor}
\usepackage{bigstrut}
\usepackage{array}
\newcolumntype{C}[1]{>{\centering\arraybackslash}p{#1}}
\usepackage{footnote}
\usepackage[demo]{graphicx}
\usepackage{threeparttable,tablefootnote}
\usepackage[titletoc,toc,title]{appendix}
\usepackage[round]{natbib}
\usepackage{authblk}
\usepackage{enumitem}
\usepackage{hhline}

\newcommand\numberthis{\addtocounter{equation}{1}\tag{\theequation}}

\usepackage{fancyhdr}
\usepackage{tikz}

\theoremstyle{theorem}
\newtheorem{remark}{Remark}
\newtheorem{thm}{Theorem}
\newtheorem{assump}{Assumption}

\makeatletter
\renewcommand\@biblabel[1]{}

\makeatother

\makeatletter
\def\old@comma{,}
\catcode`\,=13
\def,{%
	\ifmmode%
	\old@comma\discretionary{}{}{}%
	\else%
	\old@comma
	\fi%
}
\makeatother

\usepackage{setspace}
\begin{document}
	
	\begin{center} {\bf\large On the Choice of Instruments in Mixed Frequency Specification Tests}\end{center}
	\centerline{\textsc{Yun Liu$^1$%
			and Yeonwoo Rho$^{2}$\footnote[0]{
				Address for correspondence: Yeonwoo Rho, Department of Mathematical Sciences, Michigan Technological University, Houghton, MI 49931, USA. (Email: yrho@mtu.edu)}}}
	\centerline {\textsuperscript{1,2}\it  Michigan Technological University}
	\bigskip
	\centerline{\today}
	\bigskip

	\begin{abstract}
		Time averaging has been the traditional approach to handle mixed sampling frequencies. However, it ignores information possibly embedded in high frequency. Mixed data sampling (MIDAS) regression models provide a concise way to utilize the additional information in high-frequency variables. In this paper, we propose a specification test to choose between time averaging and MIDAS models, based on a Durbin-Wu-Hausman test. In particular, a set of instrumental variables is proposed and theoretically validated when the frequency ratio is large. As a result, our method tends to be more powerful than existing methods, as reconfirmed through the simulations.\\
		\textit{key words}: Mixed data sampling regression model; instrumental variable; Durbin-Wu-Hausman test; specification test; time averaging.
	\end{abstract}

	\section{Introduction}\label{introduction}
	
	In recent years, datasets that involve different sampling frequencies have drawn substantial attention in various fields. 
	Several methods were introduced to handle mixed-frequency variables in a regression model.
	One conventional approach is time averaging of high-frequency variables, where high-frequency variables are aggregated using a predetermined fixed-weight function. 
	Another is the autoregressive distributed lag (ADL) model, in which all high-frequency variables are used as regressors. 
	The mixed data sampling (MIDAS) regression model \citep{2004} was proposed to balance the complexity and the flexibility of these two approaches. In MIDAS models, the weight function is written as a nonlinear parametric function with a few parameters. The elements in the weight function do not move as freely as the ones in the ADL model due to the parametric restriction. They are still more flexible than those in time averaging since  parameters in the weight function are determined by data.
	This idea of concise yet data-driven reduction of information embedded in high sampling frequency has driven a recent surge of interest in MIDAS models \citep{foroni2013survey}.
	
	However, MIDAS models involve nonlinear estimation.  
	If the time averaging is good enough, there is no need to go through this nonlinear estimation.
	This motivates a specification test that helps decide between the  time averaging and the MIDAS models.
	There have yet been only a handful of such tests.
	\citet{andreou} presented a Durbin-Wu-Hausman (DWH) type test, designed to see whether there is an omitted variable bias  caused by overlooking the MIDAS effect. 
	\citet{Miller2016} presented two variable addition test (VAT) statistics. 
	In particular, the second VAT statistic, called a modified VAT statistic, was designed for nonstationary high-frequency variables. 
	\cite{Groenvik:Rho:2017} further extended Miller's first VAT statistic using a self-normalized approach. 
	
	In this paper,
	we shall further explore the DWH specification test introduced in \cite{andreou}. 
	In particular, the DWH test requires choosing appropriate instrumental variables, but there has not yet been a practical guidance so far. 
	We shall propose a set of instrumental variables that is suitable for this test.
	Section \ref{methodpart} presents details of such a choice, demonstrating its theoretical consistency when the frequency ratio is large enough.
	Section \ref{MC-S} presents finite sample comparisons. 
	All technical proofs  and full simulation results can be found in the appendix. 
	
	The following notations are used consistently throughout the manuscript.
	Let $T$ be the sample size at low frequency, and $m$ be the frequency ratio between the two sampling frequencies.
	$\mathbf{j}_t$ is a $T\times1$ vector with the $t$-th element being 1 and the rest 0. 
	$\mathbf{j}$ is a $T\times1$ vector of 1's. 
	Symbols $\mathbf{y}=(y_{1},\cdots,y_{T})'$,  $\mathbf{x}_{t,m}^{(m)}=\left(x_t,x_{t-1/m},\cdots,x_{t-(m-1)/m}\right)'$, 
	and
	$\mathbf{z}_t=(z_{1,t},\cdots,z_{p,t})'$
	are reserved for the low frequency variable, the high frequency variable, and  $p$ instrumental variables, respectively.
	We  use
	$\boldsymbol{\pi}=(\pi_1, \cdots, \pi_m)'$ to indicate  an $m\times1$ weight vector to aggregate the high frequency variable such that $\pi_i\geq0$ and $\sum_{i=1}^m \pi_i=1$. 
	The matrix $\mathbf{P}_{\mathbf{A}}=\mathbf{A}(\mathbf{A}'\mathbf{A})^{-1}\mathbf{A}'$ denotes the projection matrix onto the space spanned by the columns of $\mathbf{A}$,
  and 
   $\mathbf{M}_{\mathbf{A}}=\mathbf{I}-\mathbf{P}_{\mathbf{A}}$.
	For convenience, we  define the following matrices:
	$\mathbf{X}=[\mathbf{x}_{1,m}^{(m)},\cdots,\mathbf{x}_{T,m}^{(m)}]'$,
	$\mathbf{Z}=[\mathbf{z}_1,\cdots,\mathbf{z}_T]'$, and
	$\mathbf{X}^A=[\mathbf{j},\mathbf{X}\boldsymbol{\pi}_0]=[\mathbf{x}_1^A,\cdots,\mathbf{x}_T^A]'$, where $\mathbf{x}_t^A=(1,x_t^A)'$ is the $t$-th row of $\mathbf{X}^A$ and $\boldsymbol{\pi}_0$ is the predetermined weight vector.
  
	\section{Choice of Instrumental Variables Based on the DWH Test}\label{methodpart}
Consider a dataset with different sampling frequencies. 
Let  $\{y_t\}_{t=1}^T$ and $\{\mathbf{x}_{t,m}^{(m)}\}_{t=1}^T$
be the variables observed at lower and higher sampling frequencies, respectively.
The MIDAS model is constructed, aiming to model low-frequency variable using high-frequency variable:
\begin{equation}\label{MIDAS}
y_{t}=\beta_0+\left(\mathbf{j}'_t\mathbf{X}\boldsymbol{\pi}(\boldsymbol{\theta})\right)\beta_1+u_{t}
,~~~t=1,\ldots,T.
\end{equation}
The error process $\{u_{t}\}$ is stationary and uncorrelated with $\{\mathbf{x}_{t,m}^{(m)}\}$.
The vector $\boldsymbol{\pi}(\boldsymbol{\theta})=(\pi_1(\boldsymbol{\theta}),\ldots,\pi_m(\boldsymbol{\theta}))'$ consists of a function of a finite dimensional unknown parameter $\boldsymbol{\theta}$ such that  $\pi_i(\boldsymbol{\theta})\geq0$ and $\sum_{i=1}^m\pi_i(\boldsymbol{\theta})=1$.
This vector dictates how much weight would be assigned when aggregating the high-frequency variable, $\mathbf{x}_{t,m}^{(m)}$.

In a time averaging model, $\boldsymbol{\pi}=\boldsymbol{\pi}_0$ is a predetermined fixed-weight vector that does not depend on any unknown parameter $\boldsymbol{\theta}$. Without loss of generality, let the number of aggregated lags be the same as the frequency ratio $m$. Then the regression model (\ref{MIDAS}) becomes
\begin{equation}\label{FLAT}
y_{t}=\beta_0^A+\left(\mathbf{j}'_t\mathbf{X}\boldsymbol{\pi}_0\right)\beta_1^A+u_{t}^A=\beta_0^A+x_t^A\beta_1^A+u_{t}^A.
\end{equation}
We consider the test between time averaging (\ref{FLAT}) and MIDAS aggregation (\ref{MIDAS}), i.e. $H_0:\boldsymbol\pi=\boldsymbol\pi_0$ versus $H_a:\boldsymbol\pi=\boldsymbol\pi(\theta)$. 
The two commonly used weights for time averaging are the flat aggregation $\boldsymbol{\pi}_0=(1/m,\ldots,1/m)'$ and the end-of-period sampling $\boldsymbol{\pi}_0=(1,0,\ldots,0)'$. In this article, a  more general scenario of the end-of-period sampling is considered: a fixed number, $n$, of elements in $\boldsymbol{\pi}_0$ are assigned with positive values, where $n$ is independent of $m$. For brevity, we assign the first $n$ elements and leave the rest as zero, i.e.  $\boldsymbol{\pi}_0=(\pi_{0,1},\ldots,\pi_{0,n},0,\ldots,0)'$ where $\pi_{0,i}>0$ for $i=1,\cdots,n$ and $\sum_{i=1}^{n}\pi_{0,i}=1$. 

The least squares  (LS) principle can be applied to estimate the parameters $\beta_0^A$ and $\beta_1^A$ in model (\ref{FLAT}) when the null hypothesis is true. We call this estimator, $\widehat{\boldsymbol\beta}^A=(\widehat{\beta}_0^A,\widehat{\beta}_1^A)'=({\mathbf{X}^{A}}'\mathbf{X}^{A})^{-1}{\mathbf{X}^{A}}'\mathbf{y}$,  the NULL-LS estimator. 
By comparing models (\ref{MIDAS}) and (\ref{FLAT}),
the error process (\ref{FLAT}) can be rewritten as $u_{t}^A=u_t+\mathbf{j}'_t\mathbf{X}\left(\boldsymbol{\pi}(\boldsymbol{\theta})-\boldsymbol{\pi}_0\right)\beta_1$.
Under the null, $u_{t}^A$ is uncorrelated with $x_t^A$ since $u_t^A=u_t$.
However, under the alternative, $u_{t}^A$ is correlated with $x_t^A$ due to the omitted variable. 
Therefore, testing whether $\boldsymbol\pi=\boldsymbol\pi_0$ is equivalent to testing whether the NULL-LS estimator is consistent. 

To test the consistency of the NULL-LS estimator using a DWH-type test, another estimator that is consistent under both the null and the alternative is required. This estimator may not be efficient under the null. See \citet{Lee:2010}, for example. 
The two stage least squares (2SLS) estimator with proper instruments could be such an estimator. 
Assume that the instruments $\mathbf{z}_t$ are correlated with $x_t^A$, but uncorrelated with $u_{t}^A$. 
Consider a two stage regression model: the time-averaging model (\ref{FLAT}) and an auxiliary regression of the flat aggregated term $x_t^A$ on the instrumental variable $\mathbf{z}_t$ given as
\begin{equation}\label{firstR}
y_{t}=\beta_0+x_t^A\beta_1 +u_{t}^A\ \ \  \text{and}\ \ \ x_t^A=\mathbf{z}'_t\Gamma+\varepsilon_t,
\end{equation}
where $E\left(\varepsilon_t|x_t^A\right)=0$. 
The 2SLS estimator is $
\widehat{\boldsymbol\beta}=({\mathbf{X}^A}'\mathbf{P}_{\mathbf{Z}}\mathbf{X}^A)^{-1}({\mathbf{X}^A}'\mathbf{P}_{\mathbf{Z}}\mathbf{y})$. 
The bias of the 2SLS estimator $\widehat{\boldsymbol\beta}$ of $\boldsymbol\beta$ can be written as
\begin{equation}\label{dif}
\begin{aligned}
\widehat{\boldsymbol\beta}-\boldsymbol\beta
&=({\mathbf{X}^A}'\mathbf{P}_{\mathbf{Z}}{\mathbf{X}^A})^{-1}({\mathbf{X}^A}'\mathbf{P}_{\mathbf{Z}}){\mathbf{u}^A},
\end{aligned}
\end{equation}
where $\mathbf{u}^A=(u_1^A,\ldots,u_T^A)'$.
The following Assumption \ref{a1} is for the consistency of the NULL-LS under the null and for  the consistency of the 2SLS estimator under both the null and the alternative. 
	
	\begin{assump}\label{a1}
		Consider  the time-averaging model and the auxiliary regression in (\ref{firstR}).
		\begin{enumerate}[label=(\alph*)]
			\item  $T^{-1}{\mathbf{X}^A}'\mathbf{X}^A\xrightarrow{p}E\left({\mathbf{x}_t^A}{\mathbf{x}_t^A}'\right)=\mathbf{Q}_{XX}$ for some positive definite matrix $\mathbf{Q}_{XX}$; \label{a1a}
			
			\item $T^{1/2}\left(T^{-1}{\mathbf{X}^A}'\mathbf{u}^A-E\left({\mathbf{x}_t^A}{{u}_t^A}\right)\right)\xrightarrow{d}N(\mathbf{0},\boldsymbol\Omega)$ for some matrix $\boldsymbol\Omega$. Under the null, $E\left({\mathbf{x}_t^A}{{u}_t^A}\right)=0$; \label{a1b}
			
			\item Rank of $\mathbf{Z}$ is no less than the column rank of $\mathbf{X}^A$; \label{a1c}
			
			\item $T^{-1}\mathbf{Z}'\mathbf{Z}\xrightarrow{p}E\left(\mathbf{z}_t\mathbf{z}'_t\right)=\mathbf{Q}_{ZZ}$ for some positive definite matrix $\mathbf{Q}_{ZZ}$; \label{a1d}

			\item$T^{-1}{\mathbf{X}^A}'\mathbf{Z}\xrightarrow{p}E\left({\mathbf{x}_t^A}\mathbf{z}_t'\right)=\mathbf{Q}_{XZ}$ for some positive definite matrix $\mathbf{Q}_{XZ}$ with rank as the column rank of $\mathbf{X}^A$; \label{a1e}
			
			\item $T^{-1}{\mathbf{Z}}'\mathbf{u}^A\xrightarrow{p}E\left({\mathbf{z}_t}{u}_t^A\right)=\mathbf{0}$; \label{a1f}
			
			\item $T^{-1/2}{\mathbf{Z}}'\mathbf{u}^A\xrightarrow{d}N(\mathbf{0},\boldsymbol\Sigma_{Zu})$ for some positive definite matrix $\boldsymbol\Sigma_{Zu}$. \label{a1g}
			
		\end{enumerate}
	\end{assump}
	
	Assumptions \ref{a1}\ref{a1a} and \ref{a1}\ref{a1b} ensure the consistency of the NULL-LS estimator. 
	Assumption \ref{a1}\ref{a1a} indicates that $\mathbf{X}^A$ has full column rank.
	Assumption \ref{a1}\ref{a1b} implies the relation between the time-averaging term $\mathbf{X}^A$ and the error process $\mathbf{u}^A$, and their product should be asymptotically normal.
	Under the null, $\mathbf{X}^A$ and $\mathbf{u}^A$ should not be correlated, leading $E\left({\mathbf{x}_t^A}{{u}_t^A}\right)=0$. Under the alternative, $\mathbf{X}^A$ and $\mathbf{u}^A$ are allowed to be correlated, i.e., $E\left({\mathbf{x}_t^A}{{u}_t^A}\right)\neq\mathbf{0}$. 
	The variance-covariance matrix $\boldsymbol{\Omega}$ in Assumption \ref{a1}\ref{a1b} can be consistently estimated. 
	This can be done, for example, using  heteroskedasticity and autocorrelation consistent (HAC) estimators \citep{Newey:West:1987,Andrews:1991}. 
	Assumptions \ref{a1}\ref{a1d}--\ref{a1g} hold under both hypotheses. These ensure the consistency of the 2SLS estimator. 
	In particular, Assumption \ref{a1}\ref{a1d} requires that $\mathbf{Z}$ and $\mathbf{u}^A$ should be uncorrelated. 
It is worth noting that  the number of instrumental variables should be greater than or equal to the rank of $\mathbf{X}^A$.  
	Refer to \citet{Ruud:2000} for more details and explanations. 
	
	
		Now we derive our test statistic. If Assumption \ref{a1} holds, the asymptotic distributions of $\widehat{\boldsymbol{\beta}}^A$ under the null and $\widehat{\boldsymbol{\beta}}$  under both hypotheses can be written as followings:
	\begin{equation}
	\sqrt{T}(\widehat{\boldsymbol{\beta}}^A-\boldsymbol{\beta})\xrightarrow{d}N(\mathbf{0},\mathbf{V}^A)~~{\rm under}~H_0~~~{\rm and}~~~\sqrt{T}(\widehat{\boldsymbol{\beta}}-\boldsymbol{\beta})\xrightarrow{d}N(\mathbf{0},\mathbf{V})~~{\rm under}~H_0~{\rm and}~H_a,
	\end{equation}
	where
	$\mathbf{V}^A=\mathbf{Q}_{XX}^{-1}\boldsymbol\Omega\mathbf{Q}_{XX}^{-1}$ and 
	$\mathbf{V}=\left(\mathbf{Q}_{XZ}\mathbf{Q}_{ZZ}^{-1}\mathbf{Q}_{XZ}'\right)^{-1}\left(\mathbf{Q}_{XZ}\mathbf{Q}_{ZZ}^{-1}\boldsymbol\Sigma_{Zu}\mathbf{Q}_{ZZ}^{-1}\mathbf{Q}_{XZ}'\right)\left(\mathbf{Q}_{XZ}\mathbf{Q}_{ZZ}^{-1}\mathbf{Q}_{XZ}'\right)^{-1}$.		
		Since both $\widehat{\boldsymbol\beta}^A$ and $\widehat{\boldsymbol\beta}$ are consistent under the null, the difference between the two estimators, $\widehat{\Delta}=\widehat{\boldsymbol\beta}-\widehat{\boldsymbol\beta}^A$ converges to zero in probability. The main idea of the DWH test is to test whether $\widehat{\Delta}$ is significantly different from $\mathbf{0}$. This is equivalent to test whether
		${\mathbf{X}^A}'\mathbf{P}_{\mathbf{Z}}\mathbf{M}_{\mathbf{X}^A}\mathbf{y}$ is significantly different from   $\mathbf{0}$, since
		$\widehat{\Delta}$ can be written as
		$\widehat{\Delta}=\widehat{\boldsymbol\beta}-\widehat{\boldsymbol\beta}^A
		=({\mathbf{X}^A}'\mathbf{P}_{\mathbf{Z}}\mathbf{X}^A)^{-1}({\mathbf{X}^A}'\mathbf{P}_{\mathbf{Z}}\mathbf{M}_{\mathbf{X}^A}\mathbf{y})$
		and $({\mathbf{X}^A}'\mathbf{P}_{\mathbf{Z}}\mathbf{X}^A)^{-1}$ is positive definite.

			 We can easily see that
		\begin{equation}\label{relation}
		\mathbf{P}_{\mathbf{Z}}\mathbf{Z}=\mathbf{Z},\ \ \mathbf{M}_{\mathbf{Z}}\mathbf{Z}=\mathbf{0},\ \ \mathbf{M}_{\mathbf{X}^A}	\mathbf{X}\boldsymbol{\pi}_0=\mathbf{0},\ \ \mathbf{j}'\mathbf{M}_{\mathbf{X}^A}\mathbf{y}=0,~~{\rm and}
		\end{equation}
		\begin{equation}\label{mid}
		{\mathbf{X}^A}'\mathbf{P}_{\mathbf{Z}}\mathbf{M}_{\mathbf{X}^A}\mathbf{y}=\left[\mathbf{j},~	\mathbf{X}\boldsymbol{\pi}_0\right]'\mathbf{P}_{\mathbf{Z}}\mathbf{M}_{\mathbf{X}^A}\mathbf{y}={\left(0, ~(\mathbf{X}\boldsymbol{\pi}_0)'\mathbf{P}_{\mathbf{Z}}\mathbf{M}_{\mathbf{X}^A}\mathbf{y}
			\right)}'.
		\end{equation}
		Thus, $(\mathbf{X}\boldsymbol{\pi}_0)'\mathbf{P}_{\mathbf{Z}}\mathbf{M}_{\mathbf{X}^A}\mathbf{y}$ should be approximately zero under the null. 
		Let $\widehat{\boldsymbol{\varepsilon}}=\mathbf{M}_{\mathbf{Z}}\mathbf{X}\boldsymbol{\pi}_0$ and $\widehat{\mathbf{u}}^A=\mathbf{M}_{\mathbf{X}^A}\mathbf{y}$ indicate the fitted residuals from (\ref{firstR}). 
Consider a regression model $\widehat{\mathbf{u}}^A={\mathbf{X}^A}\boldsymbol\alpha+{\widehat{\boldsymbol\varepsilon}}\delta+\boldsymbol\upsilon$.  Applying Frisch$-$Waugh$-$Lovell (FWL) theorem, the OLS estimator $\widehat{\delta}$ of $\delta$  is 
		\begin{equation}
		\widehat{\delta}=\left[(\mathbf{M}_{\mathbf{X}^A}\mathbf{M}_{\mathbf{Z}}\mathbf{X}\boldsymbol{\pi}_0)'(\mathbf{M}_{\mathbf{X}^A}\mathbf{M}_{\mathbf{Z}}\mathbf{X}\boldsymbol{\pi}_0)\right]^{-1}(\mathbf{M}_{\mathbf{X}^A}\mathbf{M}_{\mathbf{Z}}\mathbf{X}\boldsymbol{\pi}_0)'\mathbf{M}_{\mathbf{X}^A}\widehat{\mathbf{u}}^A.
		\end{equation}
		Note that the latter part of $\widehat{\delta}$ can be derived as $(\mathbf{M}_{\mathbf{X}^A}\mathbf{M}_{\mathbf{Z}}\mathbf{X}\boldsymbol{\pi}_0)'\mathbf{M}_{\mathbf{X}^A}\widehat{\mathbf{u}}^A=(\mathbf{X}\boldsymbol{\pi}_0)'\mathbf{M}_{\mathbf{X}^A}{\mathbf{y}}-(\mathbf{X}\boldsymbol{\pi}_0)'\mathbf{P}_{\mathbf{Z}}\mathbf{M}_{\mathbf{X}^A}{\mathbf{y}}$. Since the third relation shown in (\ref{relation}) indicates that $(\mathbf{X}\boldsymbol{\pi}_0)'\mathbf{M}_{\mathbf{X}^A}{\mathbf{y}}=0$, $\widehat{\delta}=0$ is equivalent to  $(\mathbf{X}\boldsymbol{\pi}_0)'\mathbf{P}_{\mathbf{Z}}\mathbf{M}_{\mathbf{X}^A}\mathbf{y}=0$. Hence,
		testing whether $\widehat{\Delta}$ approaches to zero in probability can be viewed as testing if the coefficient $\widehat{\delta}$ is significantly different from zero.
	Consider the test statistic
			\begin{equation}	\lambda_T=T\widehat{{\delta}}'\left(\mathbf{b}'(\widehat{\mathbf{V}}-\widehat{\mathbf{V}}^A)\mathbf{b}\right)^{-1}\widehat{{\delta}},
			\end{equation}
			where $\mathbf{b}'=-\left[(\mathbf{M}_{\mathbf{X}^A}\mathbf{M}_{\mathbf{Z}}\mathbf{X}\boldsymbol{\pi}_0)'(\mathbf{M}_{\mathbf{X}^A}\mathbf{M}_{\mathbf{Z}}\mathbf{X}\boldsymbol{\pi}_0)\right]^{-1}\left[(\mathbf{X}\boldsymbol{\pi}_0)'\mathbf{P}_{\mathbf{Z}}\mathbf{X}^A\right]$, and $\widehat{\mathbf{V}}$ and $\widehat{\mathbf{V}}^A$ are consistent estimators of $\mathbf{V}$ and $\mathbf{V}^A$, respectively. 
		\begin{thm}\label{t1}
			
			Suppose Assumption \ref{a1}  holds. Under the null hypothesis,  $\lambda_T\xrightarrow{d}\chi^2_1$.
		\end{thm}

The proof of Theorem \ref{t1} is presented in Appendix \ref{S2}. 
It is worth noting that Assumption 1 holds only when the instruments $\mathbf{z}_t$ are chosen carefully. More specifically, $\mathbf{z}_t$ should be correlated with the time-averaging term, $x_t^A$, but uncorrelated with $u_{t}^A$. This is to ensure Assumptions \ref{a1}\ref{a1e} and \ref{a1}\ref{a1f}.
Otherwise, the consistency of the 2SLS estimator may not be guaranteed. 
	However, in practice, it is difficult to find such instruments. 
	\citet{andreou} suggested using all or part of high-frequency variables as instruments.
	However, they did not provide any practical guidance that is theoretically supported.
	In fact, with their suggested choice of instruments, it is possible that the chosen instruments are correlated with the error process. 
	In this case,  the 2SLS estimators would not be consistent, which may lower the power. 
	In what follows, we shall propose a set of instruments that is theoretically valid for the DWH-type specification test.
	To derive theoretical properties, we assume following conditions on the instruments and the data generating process. 
	
	\begin{assump}\label{a2} 
		Consider assumptions for $k=0,1,\cdots,m-1, t=1,\cdots,T$,
		\begin{enumerate}[label=(\alph*)]
			\item The high-frequency processes $\{x_{t-k/m}\}$ and  $\{u_{t-k/m}\}$ are independently, identically distributed (i.i.d.) or follow stationary AR(1) processes with finite second moment respectively; \label{a2b}
			
			\item $\{u_{t-k/m}\}$ is uncorrelated with  $\{x_{t-k/m}\}$;\label{a2c}
			
			\item Suppose $\mathbf{u}_{t,m}^{(m)}=(u_t,u_{t-1/m}, \cdots, u_{t-(m-1)/m})'$ with mean zero and positive definite covariance matrix, the error process $\{u_t\}$ is an aggregated term of $\mathbf{u}_{t,m}^{(m)}$ with the weight vector $\boldsymbol{\pi}({\theta})=(\pi_1(\theta),\cdots,\pi_m(\theta))'$, i.e.,  $u_t={\mathbf{u}_{t,m}^{(m)}}'\boldsymbol\pi(\theta)$  where $\pi_j(\theta)=(2-j/m)^{4\theta}/\sum_{i=1}^m(2-i/m)^{4\theta}$. 
			\label{a2a}
		\end{enumerate}
	\end{assump}
	
	Under Assumption \ref{a2}, the low-frequency response variable $\{y_t\}$ is viewed as an MIDAS aggregation of the underlying high-frequency true process $\{y_{t-k/m}\}$,
	where $y_{t-k/m}=\beta_0+x_{t-k/m}\beta_1+u_{t-k/m}$.
	Note that $\{y_{t-k/m}\}$ is not observed in practice.
	
	If we choose too many high-frequency lags as instruments, it might lead to a problem of a large number of weak instruments. 
	As a consequence, the 2SLS estimator may be biased towards the NULL-LS estimator. 
	The bias tends to get worse when there are more excessive number of instruments compared to the number of endogenous regressors. 
	A brief explanation is presented by \citet{Greene}. 
	Based on the number of the parameters in (\ref{firstR}) and the consideration on possibly weak instruments, 
	we shall construct $p=2$ instrumental variables, $\mathbf{z}_t=(z_{1,t},z_{2,t})',\ t=1\cdots,T$, as linear combinations of the high-frequency regressor.
	Inspired by \citet{Miller2016}, we propose to choose weights of the instruments $\mathbf{z}_t$ as the following two decreasing sequences:
	\begin{equation}\label{Upsilon}
	\begin{aligned}
	\Upsilon_1&=(f_1(1), f_1(2), \cdots, f_1(m))',\ \text{where }f_1(j)=\dfrac{0.9^{j-1}}{\sum_{i=1}^{m}0.9^{i-1}},~~{\rm and}\\
	\Upsilon_2&=(f_2(1), f_2(2), \cdots, f_2(m))',\ \text{where }f_2(j)=\dfrac{m+1-j}{\sum_{i=1}^{m}(m+1-i)}.
	\end{aligned}
	\end{equation}
	It is worth noting that these weights are designed to decrease exponentially and linearly fast.
	This is to  mimic the behaviors of the MIDAS weights with exponential Almon lag and beta polynomials.
	Then the two instrumental variables can be written in a vector form as  $\mathbf{z}'_t={\mathbf{x}_{t,m}^{(m)}}'\Upsilon$, where  $\Upsilon=[\Upsilon_{1},\hspace{0.07cm}\Upsilon_{2}]$.
	The following theorem demonstrates that the proposed instruments are approximately valid when the frequency ratio is large.

	\begin{thm}\label{p1}
		Let $\mathbf{Z}_r=\mathbf{X}\Upsilon_r=(z_{r,1},\cdots,z_{r,T})'$ for $r=1,2$, where  $\Upsilon_r$ be as presented in (\ref{Upsilon}), be the two instrumental variables. Assume that Assumption \ref{a2} holds.
		Write $\mathbf{Z}=[\mathbf{Z}_1,\mathbf{Z}_2]$.
		\begin{enumerate}[label=(\alph*)]
			\item Under the null hypothesis,  $\mathbf{Z}$ satisfies Assumption \ref{a1}.
			
			\item Under the alternative hypothesis, $\mathbf{Z}$  satisfies Assumptions \ref{a1}\ref{a1a}--\ref{a1e}. For any sample size $T$, Assumptions \ref{a1}\ref{a1f} and \ref{a1g} are fulfilled approximately, as the frequency ratio $m$ approaches infinity. In fact, $E(z_{r,t}u_t^A)=O(m^{-1})$ for $r=1,2$.
		\end{enumerate}
	\end{thm}
	The proof of Theorem \ref{p1} can be found in Appendix \ref{S3}. Under both the null and the alternative, it is easy to see that {${z}_{r,t}$ is correlated with ${x}_t^A$. 
The main result of Theorem \ref{p1} is that ${z}_{r,t}$ and ${u}_t^A$ are asymptotically uncorrelated when the frequency ratio is large, with the rate $E(z_{r,t}u_t^A)=O(m^{-1})$. Hence, the 2SLS estimator using our choice of the instruments is consistent when the frequency ratio $m$ is large.
	On the other hand, when $m$ is small, $T^{-1}{\mathbf{Z}}'\mathbf{u}^A$ converges, in probability, to a nonzero constant. 
	Thus, the DWH specification test with our choice of instruments would only work when $m$ is large enough.
	This explains the low power of our test in finite samples when $m$ is small in the next section.

	\section{Monte Carlo Simulations}\label{MC-S}
	
	In this section, we examine finite sample sizes and powers of our method and two other comparable methods in literature: the second test presented in \citet{andreou} (AGK, hereafter) and the unmodified VAT test in \citet{Miller2016}.
	We first briefly introduce algorithms of the methods in comparison. 
	
	\textbf{Algorithm 1} [Our Method]
	
	\begin{enumerate}
		\item Obtain $x_t^A={\mathbf{x}_{t,m}^{(m)}}'\boldsymbol\pi_0$. 
		Choose $\mathbf{z}_t={\mathbf{x}_{t,m}^{(m)}}'\Upsilon$ with $\Upsilon$ in (\ref{Upsilon}). 
		Regress $y_{t}$ on $x^A_t$ to obtain the fitted error process $\widehat{u}_{t}^A$.
		Regress $x^A_t$ on $\mathbf{z}_t$ to obtain the fitted error processes  $\widehat{\varepsilon}_t$.
		\item Regress $\widehat{u}^A_{t}$ on ${x^A_t}$ and $\widehat{\varepsilon}_t$ using $\widehat{u}^A_{t}=\alpha_0+{x^A_t}\alpha+\widehat{\varepsilon}_t\delta+\upsilon_t$. Test if the LS estimator $\widehat{\delta}$ of $\delta$ is significantly different from zero using a $t$ test.
		The standard error is calculated using a heteroscedasticity and autocorrelation consistent (HAC) estimator \citep{Newey:West:1987,Andrews:1991}.
		
	\end{enumerate}

	\textbf{Algorithm 2} [Miller's Method]
	
	\begin{enumerate}
		\item Obtain $x_t^A={\mathbf{x}_{t,m}^{(m)}}'\boldsymbol\pi_0$. Choose $\mathbf{z}_t={\mathbf{x}_{t,m}^{(m)}}'\Upsilon$ with $\Upsilon$ in (\ref{Upsilon}). Regress $y_{t}$ on $x^A_{t}$, and obtain the fitted residual $\widehat{u}_{t}^A$.\label{1}
		\item Regress $\widehat{u}_{t}^A$ on ${x^A_t}$ and $\mathbf{z}_t$ 
		using $\widehat{u}^A_{t}=\alpha_0+{x^A_t}\alpha+\mathbf{z}_t'\phi+\upsilon_t$. 
		Test if the LS estimator $\widehat{\phi}$ of $\phi$ is significantly different from zero using a Wald statistic and a HAC covariance estimator.
	\end{enumerate}
	
	\begin{remark}{\rm 
			The AGK method can also be implemented using Algorithm 1. 
			To limit the number of instruments, the first two regressors of the high-frequency variable are used in our simulations.
	}\end{remark}
	
	\begin{remark}\label{remark2}{\rm 
			It is worth noting that our method and Miller's unmodified VAT are  similar.
			Both methods utilize the two MIDAS-type aggregations, $\mathbf{z}_t$, of the high-frequency variable.
			While our method uses $\mathbf{z}_t$ as instruments under the classical framework with omitted variables, 
			Miller's use of $\mathbf{z}_t$ is more direct. 
			Miller's method searches whether the elements of $\mathbf{z}_t$  have any significant  effect on residual of $y_t$ after taking time averaging into account.
	}\end{remark}
		
	To make the results comparable, we use a simulation setting similar to the one proposed by \citet{Miller2016}.
	At high-frequency level, data are generated with $y_{t-k/m}=x_{t-k/m}\beta+u_{t-k/m}$ for $t=1,\ldots,T$, $k=0,\ldots,m-1$. 
	The high-frequency processes $\{x_{t-k/m}\}$ and $\{u_{t-k/m}\}$  are generated  as  stationary AR(1) processes given by $u_{t-k/m}=cu_{t-(k+1)/m}+\eta_{t-k/m}$ and $x_{t-k/m}=dx_{t-(k+1)/m}+\eta^*_{t-k/m}$, where $\{\eta_{t-k/m}\}$ and $\{\eta^*_{t-k/m}\}$ are i.i.d. $N(0,1)$. 
	Let $\beta=10$. Denote $\mathbf{y}_{t,m}^{(m)}=(y_t,y_{t-1/m}, \cdots, y_{t-(m-1)/m})'$ and $\mathbf{u}_{t,m}^{(m)}$ be the unobserved high-frequency response and the error process between time $t-1$ and $t$.
	Let $\boldsymbol{\pi}_0=\mathbf{j}/m$ and  $\boldsymbol{\pi}(\theta)=(\pi_1(\theta),\cdots,\pi_m(\theta))$, where $\pi_j(\theta)$ is defined in  Assumption \ref{a2}\ref{a2a}. 
	The low-frequency processes are generated as $y_{t}={\mathbf{y}_{t,m}^{(m)}}'\boldsymbol\pi(\theta)$ and $u_{t}={\mathbf{u}_{t,m}^{(m)}}'\boldsymbol\pi(\theta)$. 
	Here, $\theta=\theta_0=0$ indicates the flat aggregation, which corresponds to the null.
	If $\theta\neq0$, the weights are no longer flat. 
	Let $\theta=\theta_0+k$ where $k\in\{0.1,0.2,\cdots,1.9,2.0\}$ represent MIDAS-type alternatives. 
	The nominal level is $0.05$. 
	$R=2000$ Monte Carlo replications are generated. 
	The sample sizes is $T\in\{125,512\}$. 
	The frequency ratio is $m\in\{4,150,365\}$.

		\begin{table}[H]
			\setlength\tabcolsep{1.6pt}
			\label{TABLE}
			\centering
			\footnotesize
			\caption[Caption for LOF]{\footnotesize Empirical Sizes and Powers of our method (new), AGK, and Miller's method in the Representative Simulation Model}
			\begin{tabular}{ccccc|cccccccccccccccccccc}
				\hhline{=========================}
				\hline
				T     & m     & c     &   $k$    & 0.0   & 0.1   & 0.2   & 0.3   & 0.4   & 0.5   & 0.6   & 0.7   & 0.8   & 0.9   & 1.0   & 1.1   & 1.2   & 1.3   & 1.4   & 1.5   & 1.6   & 1.7   & 1.8   & 1.9   & 2.0 \\
				\hline
				\multirow{18}[10]{*}{{125}} & \multirow{6}[3]{*}{{4}} & \multirow{3}[1]{*}{{0.0}} & Miller & 6.3   &  100 & 100  & 100   & 100   & 100   & 100   & 100   & 100   & 100   & 100   & 100   & 100   & 100   & 100   & 100   & 100   & 100   & 100   & 100   & 100 \\
				&       &       & AGK   & 6.4   & 100   & 100   & 100   & 100   & 100   & 100   & 100   & 100   & 100   & 100   & 100   & 100   & 100   & 100   & 100   & 100   & 100   & 100   & 100   & 100 \\
				&       &       & New   & \textbf{5} &  \textbf{0} &  \textbf{0} &  \textbf{0} &  \textbf{10} &  \textbf{66} & 97  & 100  & 100   & 100   & 100   & 100   & 100   & 100   & 100   & 100   & 100   & 100   & 100   & 100   & 100 \\
				\cline{3-25}
				&       & \multirow{3}[2]{*}{{0.8}} & Miller & 7   & 99 & 100 & 100 & 100 & 100 & 100  & 100  & 100   & 100  & 100   & 100   & 100   & 100   & 100   & 100   & 100   & 100   & 100   & 100   & 100 \\
				&       &       & AGK   & 6.2   & 99    & 100   & 100   & 100   & 100   & 100   & 100   & 100   & 100   & 100   & 100   & 100   & 100   & 100   & 100   & 100   & 100   & 100   & 100   & 100 \\
				&       &       & New   & \textbf{5.3} & \textbf{4} & \textbf{1} & \textbf{4} & \textbf{14} & \textbf{46} & \textbf{81} & 97  & 100  & 100   & 100   & 100   & 100   & 100   & 100   & 100   & 100   & 100   & 100   & 100   & 100 \\
				\cline{2-25}
				& \multirow{6}[4]{*}{{150}} & \multirow{3}[2]{*}{{0.0}} & Miller & 6.4 & 100 & 100  & 100   & 100   & 100   & 100   & 100   & 100   & 100   & 100   & 100   & 100   & 100   & 100   & 100   & 100   & 100   & 100   & 100   & 100 \\
				&       &       & AGK   & 6     & \textbf{26} & \textbf{39} & \textbf{46} & \textbf{51} & \textbf{56} & \textbf{60} & \textbf{63} & \textbf{67} & \textbf{70} & \textbf{72} & \textbf{74} & \textbf{76} & \textbf{77} & \textbf{79} & \textbf{90} & \textbf{81} & \textbf{82} & \textbf{83} & \textbf{84} & \textbf{84} \\
				&       &       & New   & \textbf{5.5}   & 100   & 100   & 100   & 100   & 100   & 100   & 100   & 100   & 100   & 100   & 100   & 100   & 100   & 100   & 100   & 100   & 100   & 100   & 100   & 100 \\
				\cline{3-25}
				&       & \multirow{3}[2]{*}{{0.8}} & Miller & \textbf{5.7}   & \textbf{72} & 100 & 100   & 100   & 100   & 100   & 100   & 100   & 100   & 100   & 100   & 100   & 100   & 100   & 100   & 100   & 100   & 100   & 100   & 100 \\
				&       &       & AGK   & \textbf{5.7} & \textbf{8} & \textbf{13} & \textbf{20} & \textbf{28} & \textbf{35} & \textbf{40} & \textbf{45} & \textbf{51} & \textbf{56} & \textbf{60} & \textbf{64} & \textbf{66} & \textbf{69} & \textbf{71} & \textbf{73} & \textbf{75} & \textbf{76} & \textbf{77} & \textbf{79} & \textbf{80} \\
				&       &       & New   & \textbf{5.7} & \textbf{73} & 100  & 100   & 100   & 100   & 100   & 100   & 100   & 100   & 100   & 100   & 100   & 100   & 100   & 100   & 100   & 100   & 100   & 100   & 100 \\
				\cline{2-25}
				& \multirow{6}[3]{*}{{365}} & \multirow{3}[2]{*}{{0.0}} & Miller & 6.6     & 100 & 100  & 100   & 100   & 100   & 100   & 100   & 100   & 100   & 100   & 100  & 100  & 100  & 100  & 100  & 100  & 100  & 100  & 100  & 100 \\
				&       &       & AGK   & 5.8   & \textbf{14} & \textbf{18} & \textbf{20} & \textbf{23} & \textbf{25} & \textbf{27} & \textbf{30} & \textbf{33} & \textbf{35} & \textbf{37} & \textbf{39} & \textbf{41} & \textbf{43} & \textbf{45} & \textbf{47} &\textbf{49} & \textbf{52} & \textbf{54} & \textbf{55} & \textbf{56} \\
				&       &       & New   & \textbf{4.6} & 100   & 100   & 100   & 100   & 100   & 100   & 100   & 100   & 100   & 100   & 100   & 100   & 100   & 100   & 100   & 100   & 100   & 100   & 100   & 100 \\
				\cline{3-25}
				&       & \multirow{3}[1]{*}{ \rotatebox[origin=c]{90}{0.8}} & Miller & 6.5 & \textbf{70} & 100  & 100   & 100   & 100   & 100   & 100   & 100   & 100   & 100   & 100   & 100   & 100   & 100   & 100   & 100   & 100   & 100   & 100   & 100 \\
				&       &       & AGK   & 5.8   & \textbf{7} & \textbf{9} & \textbf{12} & \textbf{15} & \textbf{16} & \textbf{18} & \textbf{21} & \textbf{23} & \textbf{27} & \textbf{28} & \textbf{30} & \textbf{32} & \textbf{35} & \textbf{37} & \textbf{39} & \textbf{41} & \textbf{43} & \textbf{45} & \textbf{47} & \textbf{48} \\
				&       &       & New   & \textbf{4.6}   & \textbf{76} & 100  & 100   & 100   & 100   & 100   & 100   & 100   & 100   & 100   & 100   & 100   & 100   & 100   & 100   & 100   & 100   & 100   & 100   & 100 \\
				\hline
				\multirow{18}[10]{*}{{512}} & \multirow{6}[3]{*}{{4}} & \multirow{3}[1]{*}{{0.0}} & Miller & 6.3 & 100 & 100 & 100 & 100  & 100  & 100  & 100   & 100   & 100   & 100   & 100   & 100   & 100   & 100   & 100   & 100   & 100   & 100   & 100   & 100 \\
				&       &       & AGK   & \textbf{5} & 100   & 100   & 100   & 100   & 100   & 100   & 100   & 100   & 100   & 100   & 100   & 100   & 100   & 100   & 100   & 100   & 100   & 100   & 100   & 100 \\
				&       &       & New   & 5.4   & \textbf{6} & \textbf{5} & \textbf{3} & \textbf{1} & \textbf{0} & \textbf{0} & \textbf{0} & \textbf{0} & \textbf{1} & \textbf{3} & \textbf{14} & \textbf{42} & \textbf{79} & 96  & 99  & 100   & 100   & 100   & 100   & 100 \\
				\cline{3-25}
				&       & \multirow{3}[2]{*}{{0.8}} & Miller & 5.6 & \textbf{62} & 100 & 100  & 100  & 100   & 100   & 100   & 100   & 100   & 100   & 100   & 100   & 100   & 100   & 100   & 100   & 100   & 100   & 100   & 100 \\
				&       &       & AGK   & \textbf{5.1} & \textbf{61} & 99  & 100   & 100   & 100   & 100   & 100   & 100   & 100   & 100   & 100   & 100   & 100   & 100   & 100   & 100   & 100   & 100   & 100   & 100 \\
				&       &       & New   & 5.4   & \textbf{6} & \textbf{6} & \textbf{5} & \textbf{4} & \textbf{3} & \textbf{2} & \textbf{3} & \textbf{4} & \textbf{5} & \textbf{9} & \textbf{14} & \textbf{24} & \textbf{40} & \textbf{57} & \textbf{75} & \textbf{88} & 96    & 99  & 100  & 100 \\
				\cline{2-25}
				& \multirow{6}[4]{*}{{150}} & \multirow{3}[2]{*}{{0.0}} & Miller & 6.4 & 100 & 100 & 100 & 100 & 100  & 100  & 100   & 100   & 100   & 100   & 100   & 100   & 100   & 100   & 100   & 100   & 100   & 100   & 100   & 100 \\
				&       &       & AGK   & \textbf{5.5}   & \textbf{22} & \textbf{53} & \textbf{75} & \textbf{85} & 90  & 93  & 94  & 95  & 96  & 96  & 97  & 97  & 97  & 98  & 98  & 98    & 98  & 98  & 98  & 98 \\
				&       &       & New   & \textbf{5.5}   & 100   & 100   & 100   & 100   & 100   & 100   & 100   & 100   & 100   & 100   & 100   & 100   & 100   & 100   & 100   & 100   & 100   & 100   & 100   & 100 \\
				\cline{3-25}
				&       & \multirow{3}[2]{*}{{0.8}} & Miller & 6.1 & \textbf{23} & \textbf{71} & 97  & 100  & 100   & 100   & 100   & 100   & 100   & 100   & 100   & 100   & 100   & 100   & 100   & 100   & 100   & 100   & 100   & 100 \\
				&       &       & AGK   & \textbf{5.1} & \textbf{6} & \textbf{8} & \textbf{12} & \textbf{17} & \textbf{24} & \textbf{31} & \textbf{38} & \textbf{46} & \textbf{53} & \textbf{60} & \textbf{67} & \textbf{72} & \textbf{77} & \textbf{80} & \textbf{84} & \textbf{86} & \textbf{89} & {90} & 91  & 92 \\
				&       &       & New   & 5.4   & \textbf{25} & \textbf{73} & 97  & 100  & 100   & 100   & 100   & 100   & 100   & 100   & 100   & 100   & 100   & 100   & 100   & 100   & 100   & 100   & 100   & 100 \\
				\cline{2-25}
				& \multirow{6}[3]{*}{{365}} & \multirow{3}[2]{*}{{0.0}} & Miller & 5.8 & 100 & 100 & 100 & 100 & 100  & 100  & 100   & 100   & 100   & 100   & 100   & 100   & 100   & 100   & 100   & 100   & 100   & 100   & 100   & 100 \\
				&       &       & AGK   & 4.7   & \textbf{10} & \textbf{23} & \textbf{35} & \textbf{45} & \textbf{52} & \textbf{56} & \textbf{59} & \textbf{62} & \textbf{65} & \textbf{67} & \textbf{69} & \textbf{70} & \textbf{71} & \textbf{73} & \textbf{73} & \textbf{74} & \textbf{75} & \textbf{76} & \textbf{77} & \textbf{78} \\
				&       &       & New   & \textbf{5.2} & 100   & 100   & 100   & 100   & 100   & 100   & 100   & 100   & 100   & 100   & 100   & 100   & 100   & 100   & 100   & 100   & 100   & 100   & 100   & 100 \\
				\cline{3-25}
				&       & \multirow{3}[1]{*}{{0.8}} & Miller & 5.6 & \textbf{24} & \textbf{71} & 97 & 100   & 100   & 100   & 100   & 100   & 100   & 100   & 100   & 100   & 100   & 100   & 100   & 100   & 100   & 100   & 100   & 100 \\
				&       &       & AGK   & \textbf{4.7}   & \textbf{5} & \textbf{5} & \textbf{7} & \textbf{8} & \textbf{11} & \textbf{12} & \textbf{16} & \textbf{19} & \textbf{23} & \textbf{26} & \textbf{30} & \textbf{33} & \textbf{36} & \textbf{40} & \textbf{42} & \textbf{45} & \textbf{47} & \textbf{51} & \textbf{53} & \textbf{55} \\
				&       &       & New   & \textbf{5.3}   & \textbf{29} & \textbf{78} & 98  & 100   & 100   & 100   & 100   & 100   & 100   & 100   & 100   & 100   & 100   & 100   & 100   & 100   & 100   & 100   & 100   & 100 \\
				\hhline{=========================}
			\end{tabular}
			\begin{tablenotes}
				\item \scriptsize All values are shown as percentage. The nominal level is 0.05. Monte Carlo replication 2000. Bold numbers for $k=0.0$ represent the rejection rates closest to 0.05 under the null. Bold cells for $k\neq0.0$ indicate the rejection rates less than 0.90 under the local alternatives. 
			\end{tablenotes}
		\end{table}


	Table 1 presents  the empirical sizes and the powers of our method, the AGK method and Miller's method when $c\in\{0,0.8\}$ and $d=0$. 
	The results of more comprehensive settings are presented in Appendix \ref{S4}, which are consistent to what we observe in Table 1.
	When $k=0$, sizes closest to 0.05 are presented in boldface. 
	In all our simulation settings, all methods seem to have reasonable sizes. 
	Our  and the AGK method tend to have more cases in which sizes are closer to 0.05, while
	Miller's unmodified VAT tends to slightly over-reject.

	When $k\neq0$, empirical rejection rates represent powers of the tests. 
	Powers less than 0.9 are shown in boldface. 
	When $m$ is small, our method is not as powerful as the AGK method or Miller's unmodified VAT. 
	These two methods have much better performance under all alternatives. 
	For $T=125$, when the high-frequency error is AR(1), our method is less powerful when the effect size is small ($k\leq0.6$), whether the high-frequency error is i.i.d. or not. 
	When $T=512$, the power of our method is not very large when the effect size is not large enough. 
	This observation is consistent with Theorem \ref{p1}. 
	When $m$ is small, the 2SLS estimator would not be consistent using the chosen instruments. 
	As a matter of fact, if $m=4$, the two weighted functions in constructing the instruments are almost identical. 
	Therefore, when $m$ is small, $m=4$, the AGK method seems to be good enough by choosing the most recent two high-frequency variables (out of four). 
	Miller's unmodified VAT is another attractive alternative when $m$ is small since it is as powerful as the AGK method. 
	
	However, when $m$ is large, the effect of a careful choice of instruments is more visible. 
	When $m$ is $150$, the power of he AGK method never exceeds 0.90 for all alternatives. 
	In the meantime,  our method tends to have higher power under almost all alternatives. 
	Miller's unmodified VAT tends to be just a little less powerful than our method for small effective sizes. 
	Additionally, as the sample size increases ($T=512$), the AGK method becomes more powerful for large local alternatives, while all three methods reduce the power when the effective sizes are small. 
	Except for a few small effect sizes with the AR(1) high-frequency error process, our method has the highest power for most cases.
	Similar conclusions can be drawn for $m=365$. 
	
	\begin{remark}{\rm 
		It is worth noting that when our method works, i.e., when $1/m$ is small enough, our method and Miller's method have similar finite sample performance,
		though our test tends to have slightly better sizes and powers. 
		Considering their similar formulation as mentioned in Remark \ref{remark2}, this similarity is somewhat expected.
		If one is interested in comparison between the two methods,  it would be interesting to consider more than one regressors.
		In this case, our method calls for more than two instruments, $\mathbf{z}_t$ would be different, making it easier to see the difference between the two methods.
		However, this is out of scope of this paper, we leave it as a future work.
	}\end{remark}

	\section{Conclusion}\label{conclusion}
	
	In this paper, we considered a DWH test to choose between the time-averaging models and MIDAS models.
	For the DWH test, the instruments need to be carefully chosen to avoid the problems involved with weak instruments and correlation with the error terms.
	However, there had not yet been a rigorous work regarding the proper choice of instruments.
	The main contribution of this paper is that a set of instruments has been proposed with a theoretical validation.
	In particular, the proposed instruments would only work when the frequency ratio is large enough.
	The Monte Carlo simulations reconfirm our theoretical findings.
	The DWH test with our proposed instruments is more powerful in finite samples compared to the one with a less careful choice of instruments.
	However, this is only the case when the frequency ratio is large enough.
	Therefore, our proposed specification test would be useful when handling two extremely different sampling frequencies such as monthly versus hourly observations.
	On the other hand, if the frequency ratio is very small, taking a few most recent high-frequency variables as the instruments or taking Miller's approach would be better.

	The main purpose of this paper is to provide an insight on a proper choice of instruments. To keep the exposition concise, we limited the scope of the paper  using somewhat strong assumptions.
	Now that we understand the behavior of the instruments better, 
	an extension of this paper to accommodate  more than one regressors and  general data generating process
	is underway.

	\section*{Acknowledgements}
	
	The authors are grateful to J. Isaac Miller and to the participants of the 48th annual meetings of Illinois Economics Association for constructive comments and suggestions. 
 \textbf{Superior}, a high-performance computing infrastructure at Michigan Technological University, was used in obtaining the Monte Carlo simulation results.
	This work was partially supported by NSF grant CPS-1739422.

	\begin{appendices}
		
		\section{Test Statistic $\lambda_T$ and Asymptotic Distribution}\label{S2}

		\begin{proof} [Proof of Theorem \ref{t1}]
			
			It is easy to see that under the null, the asymptotic distribution of $\widehat{\boldsymbol{\beta}}^A$ is $\sqrt{T}\left(\widehat{\boldsymbol{\beta}}^A-\boldsymbol{\beta}\right)\xrightarrow{d}N(\mathbf{0},\mathbf{V}^A)$. 
			Under both the null and the alternative, the asymptotic distribution $\widehat{\boldsymbol{\beta}}$ is $\sqrt{T}\left(\widehat{\boldsymbol{\beta}}-\boldsymbol{\beta}\right)\xrightarrow{d}N(\mathbf{0},\mathbf{V})$.  Moreover, for some matrix $\mathbf{V}^*$, we are able to derive $\sqrt{T}(\widehat{\boldsymbol\beta}-\widehat{\boldsymbol\beta}^A)\xrightarrow{d}N(\mathbf{0},\mathbf{V}^*)$. 
			Following the argument in Section 5.1 of \cite{Lee:2010}, 			 the asymptotic distribution of $\widehat{\boldsymbol{\Delta}}=\widehat{\boldsymbol\beta}-\widehat{\boldsymbol\beta}^A$ can be derived as
			\begin{equation}
			T\widehat{\boldsymbol{\Delta}}'\left(\widehat{\mathbf{V}}-\widehat{\mathbf{V}}^A\right)^{-1}\widehat{\boldsymbol{\Delta}}\xrightarrow{d}\chi^2_{rank\left(\mathbf{V}-\mathbf{V}^A\right)}.
			\end{equation}
			By noting that 
		$(\mathbf{X}\boldsymbol{\pi}_0)'\mathbf{P}_{\mathbf{Z}}\mathbf{M}_{\mathbf{X}^A}\mathbf{y}=(0,1)\left({\mathbf{X}^A}'\mathbf{P}_{\mathbf{Z}}\mathbf{X}^A\right)\widehat{\boldsymbol{\Delta}}$
and 
 $(\mathbf{X}\boldsymbol{\pi}_0)'\mathbf{M}_{\mathbf{X}^A}\mathbf{y}=0$,  
$\widehat{\delta}$ can be rewritten as
			\begin{equation}
			\begin{aligned}
			\widehat{\delta}&=\left[(\mathbf{M}_{\mathbf{X}^A}\mathbf{M}_{\mathbf{Z}}\mathbf{X}\boldsymbol{\pi}_0)'(\mathbf{M}_{\mathbf{X}^A}\mathbf{M}_{\mathbf{Z}}\mathbf{X}\boldsymbol{\pi}_0)\right]^{-1}\left(-(\mathbf{X}\boldsymbol{\pi}_0)'\mathbf{P}_{\mathbf{Z}}\mathbf{M}_{\mathbf{X}^A}\mathbf{y}\right)\\
			&=-\left[(\mathbf{M}_{\mathbf{X}^A}\mathbf{M}_{\mathbf{Z}}\mathbf{X}\boldsymbol{\pi}_0)'(\mathbf{M}_{\mathbf{X}^A}\mathbf{M}_{\mathbf{Z}}\mathbf{X}\boldsymbol{\pi}_0)\right]^{-1}(0,1)\left({\mathbf{X}^A}'\mathbf{P}_{\mathbf{Z}}\mathbf{X}^A\right)\widehat{\boldsymbol{\Delta}}\\
			&=-\left[(\mathbf{M}_{\mathbf{X}^A}\mathbf{M}_{\mathbf{Z}}\mathbf{X}\boldsymbol{\pi}_0)'(\mathbf{M}_{\mathbf{X}^A}\mathbf{M}_{\mathbf{Z}}\mathbf{X}\boldsymbol{\pi}_0)\right]^{-1}\left((\mathbf{X}\boldsymbol{\pi}_0)'\mathbf{P}_{\mathbf{Z}}\mathbf{X}^A\right)\widehat{\boldsymbol{\Delta}}\\
			&=\mathbf{b}'\widehat{\boldsymbol{\Delta}}.
			\end{aligned}
			\end{equation}
			where $\mathbf{b}'=-\left[(\mathbf{M}_{\mathbf{X}^A}\mathbf{M}_{\mathbf{Z}}\mathbf{X}\boldsymbol{\pi}_0)'(\mathbf{M}_{\mathbf{X}^A}\mathbf{M}_{\mathbf{Z}}\mathbf{X}\boldsymbol{\pi}_0)\right]^{-1}\left((\mathbf{X}\boldsymbol{\pi}_0)'\mathbf{P}_{\mathbf{Z}}\mathbf{X}^A\right)$.
				Thus,
			\begin{equation}
			\sqrt{T}\widehat{\delta}=\sqrt{T}\mathbf{b}'\widehat{\delta}=\sqrt{T}\left[\mathbf{b}'\left(\widehat{\boldsymbol{\beta}}-\boldsymbol{\beta}\right)-\mathbf{b}'\left(\widehat{\boldsymbol{\beta}}^A-\boldsymbol{\beta}\right)\right].
			\end{equation}
				The asymptotic distribution of $\mathbf{b}'\widehat{\boldsymbol{\beta}}^A$ is $\sqrt{T}\mathbf{b}'\left(\widehat{\boldsymbol{\beta}}^A-\boldsymbol{\beta}\right)\xrightarrow{d}N(0,\mathbf{b}'\mathbf{V}^A\mathbf{b})$ under the null. 
			The asymptotic distribution of $\mathbf{b}'\widehat{\boldsymbol{\beta}}$ is $\sqrt{T}\mathbf{b}'\left(\widehat{\boldsymbol{\beta}}-\boldsymbol{\beta}\right)\xrightarrow{d}N(0,\mathbf{b}'\mathbf{V}\mathbf{b})$ under both the null and the alternative. 
			Since the estimator $\mathbf{b}'\widehat{\boldsymbol{\beta}}^A$ is still consistent and efficient under the null, while the estimator $\mathbf{b}'\widehat{\boldsymbol{\beta}}$ is consistent under the null and the alternative, then
			\begin{equation}
			T{\left[\mathbf{b}'\left(\widehat{\boldsymbol{\beta}}-\widehat{\boldsymbol{\beta}}^A\right)\right]}'\left(\mathbf{b}'\widehat{\mathbf{V}}\mathbf{b}-\mathbf{b}'\widehat{\mathbf{V}}^A\mathbf{b}\right)^{-1}\left[\mathbf{b}'\left(\widehat{\boldsymbol{\beta}}-\widehat{\boldsymbol{\beta}}^A\right)\right]\xrightarrow{d}\chi^2_{rank\left(\mathbf{b}'(\mathbf{V}-\mathbf{V}^A)\mathbf{b}\right)}.
			\end{equation}
			Therefore,
			\begin{equation}
			T\widehat{\delta}'\left(\mathbf{b}'(\widehat{\mathbf{V}}-\widehat{\mathbf{V}}^A)\mathbf{b}\right)^{-1}\widehat{\delta}\xrightarrow{d}\chi^2_{rank\left(\mathbf{b}'(\mathbf{V}-\mathbf{V}^A)\mathbf{b}\right)}.
			\end{equation}
			
			Note that under our settings, $\mathbf{b}$ is a column vector with two elements. 
			The rank of $\mathbf{b}'(\mathbf{V}-\mathbf{V}^A)\mathbf{b}$ is one. Hence, the degree of freedom of $\chi^2$ distribution is one. 
		\end{proof}

		\section{Theoretical Verification of the Chosen Set of Instruments}\label{S3}

		\begin{proof} [Proof of Theorem \ref{p1}]
			
			It is obvious that our choice of instruments follows Assumption \ref{a1}\ref{a1c}. 
			Following Slutsky's theorem, it is straightforward to show that our choice of instruments satisfies Assumption \ref{a1}\ref{a1d} and \ref{a1}\ref{a1e}. So, the main part is to show that our choice of instruments satisfies Assumption \ref{a1}\ref{a1f}, i.e., $E(\mathbf{Z}'\mathbf{u}^A)$ is zero or approximates to zero as the frequency ratio $m$ approaches infinity. Assumption \ref{a1}\ref{a1g} follows.
			
			Under the null hypothesis, $\widehat{\boldsymbol\beta}^A$ is consistent to estimate $\boldsymbol\beta$, then the error process $\{u_t\}$ is exactly $\{u_t^A\}$ in (\ref{FLAT}). Therefore, following Assumption \ref{a2}\ref{a2c}, $	{u}^A_t={u}_t={\mathbf{u}_{t/m}^{(m)}}'\boldsymbol\pi(\theta)$, ${\mathbf{z}_t}'={\mathbf{x}_{t,m}^{(m)}}'\Upsilon$,
			\begin{equation*}
			T^{-1}{\mathbf{Z}}'\mathbf{u}^A=T^{-1}{\mathbf{Z}}'\mathbf{u}=T^{-1}\sum_{t=1}^{T}\mathbf{z}_tu_t=T^{-1}\sum_{t=1}^{T}\Upsilon'{\mathbf{x}_{t,m}^{(m)}}{\mathbf{u}_{t,m}^{(m)}}'\boldsymbol\pi(\theta)\xrightarrow{p}\mathbf{0}.
			\end{equation*}
			It follows that the asymptotic distribution is $T^{-1/2}\mathbf{Z}'\mathbf{u}^A\xrightarrow{d}N(0,\boldsymbol{\Sigma}_{Zu})$ for some matrix $\boldsymbol{\Sigma}_{Zu}$.

			Under the alternative hypothesis, $\widehat{\boldsymbol\beta}^A$ is not consistent, the true model is the MIDAS model in (\ref{MIDAS}), i.e. $\mathbf{y}={\mathbf{X}(\theta)}\boldsymbol\beta+\mathbf{u}$, 
			where $\mathbf{X}(\theta)=[\mathbf{j},\mathbf{X}\boldsymbol\pi(\theta)]$. Recall that $\mathbf{X}^A=[\mathbf{j},\mathbf{X}\boldsymbol{\pi}_0]$. Let ${\mathbf{x}_t^A}'$ and ${\mathbf{x}_t(\theta)}'$ be $t$-th row of $\mathbf{X}^A$ and $\mathbf{X}(\theta)$, respectively.
			Comparing the MIDAS model with the regression model in (\ref{FLAT}), $\mathbf{y}={\mathbf{X}^A}\boldsymbol\beta^A+\mathbf{u}^A$, 
			it is easy to show that $\boldsymbol\beta^A$ can be written as $\boldsymbol\beta^A=\left\{E\left(\mathbf{x}_t^A{\mathbf{x}_t^A}'\right)\right\}^{-1}\left\{E\left(\mathbf{x}_t^A{\mathbf{x}_t(\theta)}'\right)\right\}{\beta}$, 
			then    	
			\begin{equation}
			\begin{aligned}
			u_t^A&=y_t-{\mathbf{x}_t^A}'\boldsymbol\beta^A=y_t-{\mathbf{x}_t^A}'\left\{E\left(\mathbf{x}_t^A{\mathbf{x}_t^A}'\right)\right\}^{-1}\left\{E\left(\mathbf{x}_t^A{\mathbf{x}_t(\theta)}'\right)\right\}\boldsymbol{\beta}\\
			&=\left({\mathbf{x}_t(\theta)}'-{\mathbf{x}_t^A}'\left\{E\left(\mathbf{x}_t^A{\mathbf{x}_t^A}'\right)\right\}^{-1}\left\{E\left(\mathbf{x}_t^A{\mathbf{x}_t(\theta)}'\right)\right\}\right)\boldsymbol\beta+u_t\\
			&=\mathbf{A}\boldsymbol\beta+u_t,
			\end{aligned}
			\end{equation}
			where $\mathbf{A}={\mathbf{x}_t(\theta)}'-{\mathbf{x}_t^A}'\left\{E\left(\mathbf{x}_t^A{\mathbf{x}_t^A}'\right)\right\}^{-1}\left\{E\left(\mathbf{x}_t^A{\mathbf{x}_t(\theta)}'\right)\right\}$. 
			Let $\mathbf{J}_m=\mathbf{j}\mathbf{j}'$ be a all-ones matrix with dimension $m$. According to the property of $\boldsymbol{\pi}_0$ and $\boldsymbol{\pi}({\theta})$, we have $\boldsymbol{\pi}_0'\mathbf{j}=1$ and$\boldsymbol{\pi}({\theta})'\mathbf{j}=1$.
			
			Since the high-frequency processes $\{x_{t-k/m}\}$ and  $\{u_{t-k/m}\}$ are assumed to be i.i.d. or follow stationary AR(1) processes with finite second moment, respectively, for $k=0,1,\cdots,m-1, t=1,\cdots,T$ and $\sum_{i=1}^{m}\pi_i=1$, denote the variance-covariance matrix of $\mathbf{x}_{t,m}^{(m)}$ as $\boldsymbol\Phi=E\left({\mathbf{x}_{t,m}^{(m)}}{\mathbf{x}_{t,m}^{(m)}}'\right)-E\left({\mathbf{x}_{t,m}^{(m)}}\right)E\left({\mathbf{x}_{t,m}^{(m)}}\right)'$, then $E\left({\mathbf{x}_{t,m}^{(m)}}\right)=\mu\mathbf{j}$, $E\left({\mathbf{x}_{t,m}^{(m)}}{\mathbf{x}_{t,m}^{(m)}}'\right)=\boldsymbol{\Phi}+\mu^2\mathbf{J}_m$.   
			\begin{equation*}
			\begin{aligned}
			\mathbf{A}&=\left(\begin{matrix}
			1 & {\mathbf{x}_{t,m}^{(m)}}'\boldsymbol\pi(\theta)
			\end{matrix}\right)-\left(\begin{matrix}
			1 & {\mathbf{x}_{t,m}^{(m)}}'\boldsymbol\pi_0
			\end{matrix}\right)\left\{E\left(\mathbf{x}_t^A{\mathbf{x}_t^A}'\right)\right\}^{-1}\left\{E\left(\mathbf{x}_t^A{\mathbf{x}_t(\theta)}'\right)\right\}
			\end{aligned}
			\end{equation*}
			where
			\begin{gather*}
			E\left(\mathbf{x}_t^A{\mathbf{x}_t^A}'\right)=\left[\begin{matrix}
			1 & \boldsymbol\pi_0'E\left({\mathbf{x}_{t,m}^{(m)}}\right)\\
			\boldsymbol\pi_0'E\left({\mathbf{x}_{t,m}^{(m)}}\right) & \boldsymbol\pi_0'E\left({\mathbf{x}_{t,m}^{(m)}}{\mathbf{x}_{t,m}^{(m)}}'\right)\boldsymbol\pi_0
			\end{matrix}\right]=\left[\begin{matrix}
			1 & \mu\\
			\mu & \boldsymbol\pi_0'(\boldsymbol{\Phi}+\mu^2\mathbf{J}_m)\boldsymbol\pi_0
			\end{matrix}\right],\\
			E\left(\mathbf{x}_t^A{\mathbf{x}_t(\theta)}'\right)=\left[\begin{matrix}
			1 & \boldsymbol\pi(\theta)'E\left({\mathbf{x}_{t,m}^{(m)}}\right)\\
			\boldsymbol\pi_0'E\left({\mathbf{x}_{t,m}^{(m)}}\right) & \boldsymbol\pi_0'E\left({\mathbf{x}_{t,m}^{(m)}}{\mathbf{x}_{t,m}^{(m)}}'\right)\boldsymbol\pi(\theta)
			\end{matrix}\right]=\left[\begin{matrix}
			1 & \mu\\
			\mu & \boldsymbol\pi_0'(\boldsymbol{\Phi}+\mu^2\mathbf{J}_m)\boldsymbol\pi(\theta)
			\end{matrix}\right].
			\end{gather*}

			Assuming that $E\left(\mathbf{x}_t^A{\mathbf{x}_t^A}'\right)$ is invertible (if $E\left(\mathbf{x}_t^A{\mathbf{x}_t^A}'\right)$ is not invertible, we can get the generalized inverse), then we can derive
			\begin{equation*}
			\left\{E\left(\mathbf{x}_t^A{\mathbf{x}_t^A}'\right)\right\}^{-1}\left\{E\left(\mathbf{x}_t^A{\mathbf{x}_t(\theta)}'\right)\right\}=\left[\begin{matrix}
			1 & (\boldsymbol\pi_0'\boldsymbol\Phi\boldsymbol\pi_0)^{-1}\mu\boldsymbol{\pi}_0'(\boldsymbol{\Phi}+\mu^2\mathbf{J}_m)(\boldsymbol\pi_0-\boldsymbol\pi(\theta))\\
			0 & (\boldsymbol\pi_0'\boldsymbol\Phi\boldsymbol\pi_0)^{-1}\boldsymbol\pi_0'\boldsymbol\Phi\boldsymbol\pi(\theta)
			\end{matrix}\right],
			\end{equation*}

			Therefore,
			\begin{equation}\label{minus}
			\begin{aligned}
			\mathbf{A}&=\left(\begin{matrix}
			1 & {\mathbf{x}_{t,m}^{(m)}}'\boldsymbol\pi(\theta)
			\end{matrix}\right)-\left(\begin{matrix}
			1 & {\mathbf{x}_{t,m}^{(m)}}'\boldsymbol\pi_0
			\end{matrix}\right)\left\{E\left(\mathbf{x}_t^A{\mathbf{x}_t^A}'\right)\right\}^{-1}\left\{E\left(\mathbf{x}_t^A{\mathbf{x}_t(\theta)}'\right)\right\}\\
			&=\left(\begin{matrix}
			1 & {\mathbf{x}_{t,m}^{(m)}}'\boldsymbol\pi(\theta)
			\end{matrix}\right)-\left(\begin{matrix}
			1 & (\boldsymbol\pi_0'\boldsymbol\Phi\boldsymbol\pi_0)^{-1}\left\{\mu\boldsymbol{\pi}_0'(\boldsymbol{\Phi}+\mu^2\mathbf{J}_m)(\boldsymbol\pi_0-\boldsymbol\pi(\theta))+{\mathbf{x}_{t,m}^{(m)}}'\boldsymbol\pi_0\boldsymbol\pi_0'\boldsymbol\Phi\boldsymbol\pi(\theta)\right\}
			\end{matrix}\right)\\
			&=\left(\begin{matrix}
			0 & {\mathbf{x}_{t,m}^{(m)}}'\boldsymbol\pi(\theta)-(\boldsymbol\pi_0'\boldsymbol\Phi\boldsymbol\pi_0)^{-1}\left\{\mu\boldsymbol{\pi}_0'(\boldsymbol{\Phi}+\mu^2\mathbf{J}_m)(\boldsymbol\pi_0-\boldsymbol\pi(\theta))+{\mathbf{x}_{t,m}^{(m)}}'\boldsymbol\pi_0\boldsymbol\pi_0'\boldsymbol\Phi\boldsymbol\pi(\theta)\right\}
			\end{matrix}\right).
			\end{aligned}
			\end{equation}
			
			Next, calculate $E\left(\mathbf{z}_tu_t^A\right)$ where $\mathbf{z}'_t={\mathbf{x}_{t,m}^{(m)}}'\Upsilon$,
			\begin{equation}\label{zu}
			E\left(\mathbf{z}_tu_t^A\right)=E\left(\mathbf{z}_t(\mathbf{A}\boldsymbol\beta+u_t)\right)=E\left(\mathbf{z}_t\mathbf{A}\boldsymbol\beta\right)=E\left(\mathbf{z}_t\mathbf{A}\boldsymbol\left[\begin{matrix}
			\beta_0\\
			\beta_1
			\end{matrix}\right]\right).
			\end{equation}
			
			Combine (\ref{minus}) with (\ref{zu}), then
			\begin{equation}\label{Expectation}
			\begin{aligned}
			&E\left(\mathbf{z}_tu_t^A\right)\\
			&=\beta_1E\left(\mathbf{z}_t\left({\mathbf{x}_{t,m}^{(m)}}'\boldsymbol\pi(\theta)-(\boldsymbol\pi_0'\boldsymbol\Phi\boldsymbol\pi_0)^{-1}\left\{\mu\boldsymbol{\pi}_0'(\boldsymbol{\Phi}+\mu^2\mathbf{J}_m)(\boldsymbol\pi_0-\boldsymbol\pi(\theta))+{\mathbf{x}_{t,m}^{(m)}}'\boldsymbol\pi_0\boldsymbol\pi_0'\boldsymbol\Phi\boldsymbol\pi(\theta)\right\}\right)\right)\\
			&=\beta_1E\left(\Upsilon'{\mathbf{x}_{t,m}^{(m)}}\left({\mathbf{x}_{t,m}^{(m)}}'\boldsymbol\pi(\theta)-(\boldsymbol\pi_0'\boldsymbol\Phi\boldsymbol\pi_0)^{-1}\left\{\mu\boldsymbol{\pi}_0'(\boldsymbol{\Phi}+\mu^2\mathbf{J}_m)(\boldsymbol\pi_0-\boldsymbol\pi(\theta))+{\mathbf{x}_{t,m}^{(m)}}'\boldsymbol\pi_0\boldsymbol\pi_0'\boldsymbol\Phi\boldsymbol\pi(\theta)\right\}\right)\right)\\
			&=\beta_1\Upsilon'\left\{(\boldsymbol{\Phi}+\mu^2\mathbf{J}_m)\boldsymbol\pi(\theta)-(\boldsymbol\pi_0'\boldsymbol\Phi\boldsymbol\pi_0)^{-1}\mu\boldsymbol{\pi}_0'(\boldsymbol{\Phi}+\mu^2\mathbf{J}_m)(\boldsymbol\pi_0-\boldsymbol\pi(\theta))\mu\mathbf{j}\right.\\
			&\ \ \ \ \ \ \ \ \ \ \ \ \ \  \left.-(\boldsymbol\pi_0'\boldsymbol\Phi\boldsymbol\pi_0)^{-1}(\boldsymbol{\Phi}+\mu^2\mathbf{J}_m)\boldsymbol\pi_0\boldsymbol\pi_0'\boldsymbol\Phi\boldsymbol\pi(\theta)\right\}.
			\end{aligned}
			\end{equation}
			
			After simplification, (\ref{Expectation}) becomes
			\begin{equation}\label{main}
			\begin{aligned}
			E\left(\mathbf{z}_tu_t^A\right)&=\beta_1\Upsilon'\left(\boldsymbol{\Phi}\boldsymbol\pi(\theta)-(\boldsymbol\pi_0'\boldsymbol\Phi\boldsymbol\pi_0)^{-1}\boldsymbol\pi_0'\boldsymbol{\Phi}\boldsymbol\pi(\theta)\boldsymbol\Phi\boldsymbol\pi_0\right).
			\end{aligned}
			\end{equation}
			
			Note that  let $\pi_{0,i}$ be the $i$-th element of $\boldsymbol\pi_0$, $(\boldsymbol\Phi\boldsymbol\pi_0)_k$ be the $j$-th element of $\boldsymbol\Phi\boldsymbol\pi_0$ for $k=1,\cdots,m$, $\sigma_x^2$ be the variance of $x_{t-j/m}$ for any $t=1,\cdots,T$, $j=0,\cdots,m-1$. 
			Suppose the parameter in the high-frequency AR(1) process is $d$ such that $0<|d|<1$ (for i.i.d. case, let $d=0$ and define $0^0=1$), then we have
			\begin{equation}
			\begin{aligned}
			\boldsymbol\pi_0'\boldsymbol\Phi\boldsymbol\pi({\theta})&=\sum_{j=1}^{m}\sum_{i=1}^{m}\pi_{0,i}\phi_{i,j}\pi_{j}({\theta})=\sum_{j=1}^{m}\sum_{i=1}^{m}\pi_{0i}d^{|i-j|}\sigma_x^2\pi_{j}({\theta}),
			\\
			\boldsymbol\pi_0'\boldsymbol\Phi\boldsymbol\pi_0&=\sum_{j=1}^{m}\sum_{i=1}^{m}\pi_{0,i}\phi_{i,j}\pi_{0,j}=\sum_{j=1}^{m}\sum_{i=1}^{m}\pi_{0i}d^{|i-j|}\sigma_x^2\pi_{0j},\\
			(\boldsymbol\Phi\boldsymbol\pi_0)_k&=\sum_{j=1}^{m}d^{|k-j|}\sigma_x^2\pi_{0,j}.
			\end{aligned}
			\end{equation}
			
			As we mentioned above, the weighted matrix $\Upsilon=[\Upsilon_{1}\hspace{0.07cm}\Upsilon_{2}]$ is defined in (\ref{Upsilon}). Let $S_{\pi}=\sum_{i=1}^{m}(2-i/m)^{4\theta}$, $S_{\Upsilon_1}=\sum_{i=1}^{m}0.9^{i-1}$, $S_{\Upsilon_2}=\sum_{i=1}^{m}(m+1-i)$. 
			$\boldsymbol\pi(\theta)=(\pi_1(\theta), \cdots, \pi_m(\theta))'$, here $\pi_j(\theta)=(2-j/m)^{4\theta}/\sum_{i=1}^m(2-i/m)^{4\theta}$ for $j=1,2,\cdots,m$. 
			Consider two cases separately: (i) ${\mathbf{x}_{t,m}^{(m)}}$ is an i.i.d. sequence ($\boldsymbol\Phi=\sigma^2_x\mathbf{I}$ where $\mathbf{I}$ is the identity matrix); (ii) ${\mathbf{x}_{t,m}^{(m)}}$ is an AR(1) process with parameter $d$ where $0<|d|<1$.\\ 
			
			(i) When ${\mathbf{x}_{t,m}^{(m)}}$ is an i.i.d. sequence, then we can easily derive the following equations from (\ref{main}).
			\begin{equation}\label{case-i}
			E\left(\mathbf{z}_tu_t^A\right)=\beta_1\sigma_x^2\Upsilon'\left(\boldsymbol\pi(\theta)-\dfrac{\boldsymbol\pi_0'\boldsymbol\pi(\theta)}{\boldsymbol\pi_0'\boldsymbol\pi_0}\boldsymbol\pi_0\right)=\beta_1\sigma_x^2\Upsilon'\boldsymbol\pi(\theta)-\beta_1\sigma_x^2\left(\dfrac{\boldsymbol\pi_0'\boldsymbol\pi(\theta)}{\boldsymbol\pi_0'\boldsymbol\pi_0}\Upsilon'\boldsymbol\pi_0\right).
			\end{equation}
			
			Since $\Upsilon_r'\boldsymbol\pi(\theta)$ does not depend on the null $\boldsymbol{\pi}_0$, then we consider the first term for both the flat aggregation and the general case of end-of-period sampling. Since $\theta>0$, $S_{\pi}=O(m)$ and $1\leq(2-i/m)^{4\theta}\leq2^{4\theta}$ for $i=1,\cdots,m$, then 
			\begin{equation}
			\Upsilon_r'\boldsymbol\pi(\theta)=(S_{\pi}S_{\Upsilon_1})^{-1}\sum_{i=1}^{m}a_{i,r}(2-i/m)^{4\theta}
			\in[(S_{\pi})^{-1},2^{4\theta}(S_{\pi})^{-1}]=O(m^{-1}).
			\end{equation}
			
			Consider the time-averaging weights $\boldsymbol{\pi}_0$ with two cases respectively: (a) the flat aggregation weights $\boldsymbol{\pi}_0=(1/m,\cdots,1/m)'$; (b) $\boldsymbol{\pi}_0=(\pi_{0,1},\cdots,\pi_{0,n},0,\cdots,0)'$ for any fixed integer $n\in\left[0,m\right)$ independent of $m$ such that $\pi_{0,i}$ is positive constants independent of $m$ for all $i=1,\cdots,n$ and $\sum_{i=1}^{n}\pi_{0,i}=1$. 
			In particular, when $n=1$, it is the end-of-period sampling. 
			Note that for case (b), we can assumed that  $\boldsymbol{\pi}_0=(0,\cdots,0,\pi_{0,m-n+1},\cdots,\pi_{0,m})'$ or any fixed $n$ element with positive values of $\boldsymbol{\pi}_0$ with the property $\sum_{i=1}^{m}\pi_{0,i}=1$. 
			The proof will be straightforward by following similar processes shown below. 
			Without loss of generality, we only show the proof with the aggregating weight as $\boldsymbol{\pi}_0=(\pi_{0,1},\cdots,\pi_{0,n},0,\cdots,0)'$. 
			
			For case (a), 
			\begin{equation}
			\dfrac{\boldsymbol\pi_0'\boldsymbol\pi(\theta)}{\boldsymbol\pi_0'\boldsymbol\pi_0}=\dfrac{\sum_{i=1}^{m}\pi_{0,i}\pi({\theta})}{\sum_{i=1}^{m}\pi_{0,i}^2}=\dfrac{1/m\sum_{i=1}^{m}\pi({\theta})}{m\cdot(1/m^2)}=1,\ \ \ \Upsilon_r'\boldsymbol\pi_0=1/m,\text{\ for\ }r=1,2.
			\end{equation}
			
			Then, it follows that the second term  $\dfrac{\boldsymbol\pi_0'\boldsymbol\pi(\theta)}{\boldsymbol\pi_0'\boldsymbol\pi_0}\Upsilon_r'\boldsymbol\pi_0=O(m^{-1})$. 
			
			Hence, $E\left(\mathbf{z}_{t}u_t^A\right)=\left(O(m^{-1}),\ O(m^{-1})\right)'$\footnote{The notation $\left(O(m^{-1}),\ O(m^{-1})\right)'$ indicates that each element of this vector is equal to $O(m^{-1})$.}.
			
			For case (b), 
			\begin{equation}
			\dfrac{\boldsymbol\pi_0'\boldsymbol\pi(\theta)}{\boldsymbol\pi_0'\boldsymbol\pi_0}=\dfrac{\sum_{i=1}^{n}\pi_{0,i}\pi({\theta})}{\sum_{i=1}^{n}\pi_{0,i}^2}\leq\dfrac{(2-1/m)^{4\theta}\sum_{i=1}^{n}\pi_{0,i}}{S_{\pi}\sum_{i=1}^{n}\pi_{0,i}^2}=O(m^{-1}).
			\end{equation}
			
			\begin{equation}
			\begin{aligned}
			|\Upsilon_1'\boldsymbol\pi_0|&\leq\sigma_x^2\dfrac{\sum_{i=1}^{n}0.9^{i-1}}{S_{\Upsilon_1}}\max_{1\leq i\leq n}(\pi_{0,i})\leq\sigma_x^2\dfrac{1-0.9^n}{1-0.9^m}\leq0.1\sigma_x^2=O(1),\\
			|\Upsilon_2'\boldsymbol\pi_0|&\leq\sigma_x^2\dfrac{(m+m+1-n)n}{(m+1)m}\max_{1\leq i\leq n}(\pi_{0,i})\leq\dfrac{(2m+1-n)n}{(m+1)m}\sigma_x^2=O(m^{-1}).
			\end{aligned}
			\end{equation}
			
			It implies that the second term follows
			\begin{equation}
			\begin{aligned}
			\left|\dfrac{\boldsymbol\pi_0'\boldsymbol\pi(\theta)}{\boldsymbol\pi_0'\boldsymbol\pi_0}\Upsilon_1'\boldsymbol\pi_0\right|=O(m^{-1}),\ \ \left|\dfrac{\boldsymbol\pi_0'\boldsymbol\pi(\theta)}{\boldsymbol\pi_0'\boldsymbol\pi_0}\Upsilon_2'\boldsymbol\pi_0\right|=O(m^{-2}).
			\end{aligned}
			\end{equation}
			
			Since the first term dominantly determine the order of $E\left(\mathbf{z}_{t}u_t^A\right)$, then we can derive that $E\left(\mathbf{z}_{t}u_t^A\right)=\left(O(m^{-1}),\ O(m^{-1})\right)'$. 
			
			We have proved that with the i.i.d. high-frequency regressor, our choice of instruments satisfies Assumption \ref{a1}\ref{a1f} asymptotically in case (i). In case (ii) where the high-frequency regressor is an AR$(1)$ process, similar results can be drawn with either the flat aggregation or the end-of-period sampling in the more general scenario.\\

			(ii) When ${\mathbf{x}_{t,m}^{(m)}}$ is an AR(1) sequence with the parameter $|d|\in(0,1)$, recall (\ref{main}), 
			\begin{equation}\label{case-ii}
			\begin{aligned}
			E\left(\mathbf{z}_tu_t^A\right)&=\beta_1\Upsilon'\left(\boldsymbol{\Phi}\boldsymbol\pi(\theta)-(\boldsymbol\pi_0'\boldsymbol\Phi\boldsymbol\pi_0)^{-1}\boldsymbol\pi_0'\boldsymbol{\Phi}\boldsymbol\pi(\theta)\boldsymbol\Phi\boldsymbol\pi_0\right)=\beta_1\Upsilon'\boldsymbol{\Phi}\boldsymbol\pi(\theta)-\beta_1\Upsilon'\left(\dfrac{\boldsymbol\pi_0'\boldsymbol{\Phi}\boldsymbol\pi(\theta)}{\boldsymbol\pi_0'\boldsymbol\Phi\boldsymbol\pi_0}\boldsymbol\Phi\boldsymbol\pi_0\right).
			\end{aligned}
			\end{equation}
			
			Similar to the i.i.d. case, the first term $\boldsymbol{\Phi}\boldsymbol{\pi}({\theta})$ does not depend on the form of $\boldsymbol{\pi}_0$, then let $(\boldsymbol{\Phi}\boldsymbol{\pi}({\theta}))_k$ be the $k$-th element $\boldsymbol{\Phi}\boldsymbol{\pi}({\theta})$ for $k=1,\cdots,m$, \begin{equation}
			(\boldsymbol{\Phi}\boldsymbol{\pi}({\theta}))_k=\sigma_x^2\sum_{j=1}^{m}d^{|k-j|}\pi_j=\sigma_x^2\left(\sum_{i=k}^{m}d^{i-k}\pi_i+\sum_{j=1}^{k-1}d^{k-j}\pi_{j}\footnotemark\right).
			\end{equation}
			Note that when $k=1$, let $\sum_{j=1}^{k-1}d^j\pi_{k-j}=0$.
			\footnotetext{To simplify the notation, we will use $\pi_j$ as $j$-th element of $\boldsymbol{\pi}({\theta})$ instead of $\pi_j({\theta})$.}
			
			Recall that in (\ref{Upsilon}), we define $\Upsilon_1$ and $\Upsilon_2$ as
			\begin{equation}
			\begin{aligned}
			\Upsilon_1&=(f_1(1), f_1(2), \cdots, f_1(m))',\ \text{where }f_1(j)=0.9^{j-1}/\sum_{i=1}^{m}0.9^{i-1},~~\\
			\Upsilon_2&=(f_2(1), f_2(2), \cdots, f_2(m))',\ \text{where }f_2(j)=2(m+1-j)/\{m(m+1)\},
			\end{aligned}
			\end{equation}
			for $j=1,\cdots,m$.
			
			Since $S_{\pi}=\sum_{i=1}^{m}\pi_i=\sum_{i=1}^{m}(2-i/m)^{4\theta}\in[m,2^{4\theta}m]$, for $r=1,2$, 
			\begin{equation}\label{ineq1}
			\begin{aligned}
			\left|\Upsilon_r'\boldsymbol{\Phi}\boldsymbol{\pi}({\theta})\right|
			&=\sigma_x^2\left|\sum_{k=1}^{m}f_r(k)(\boldsymbol{\Phi}\boldsymbol{\pi}({\theta}))_k\right|
			=\sigma_x^2\left|\sum_{k=1}^{m}f_r(k)\left(\sum_{i=k}^{m}d^{i-k}\pi_i +\sum_{j=1}^{k-1}d^{k-j}\pi_{j}\right)\right|\\
			&\leq\sigma_x^2 \sum_{k=1}^{m}f_r(k)\dfrac{2^{4\theta}}{S_{\pi}}\left(\sum_{i=k}^{m}|d|^{i-k}+\sum_{j=1}^{k-1}|d|^{k-j}\right)\\
			&=\sigma_x^2\cdot\dfrac{2^{4\theta}}{S_{\pi}}\cdot\dfrac{\sum_{k=1}^{m}f_r(k)\left(1+|d|-|d|^{m-k+1}-|d|^{k}\right)}{1-|d|}\\
			&<\sigma_x^2\cdot\dfrac{2^{4\theta}}{S_{\pi}}\cdot\dfrac{\sum_{k=1}^{m}f_r(k)\left(1+|d|\right)}{1-|d|}\leq m^{-1}\sigma_x^2C_1(d,{\theta}),
			\end{aligned}
			\end{equation}
			where $C_1(d,\theta)=\dfrac{2^{4\theta}(1+|d|)}{1-|d|}$ depends on $d$ and ${\theta}$, but is independent of $m$. Therefore, the first term $\Upsilon_r'\boldsymbol{\Phi}\boldsymbol{\pi}({\theta})=O(m^{-1})$ for $r=1,2$.
			
			Consider case (a) and (b) mentioned above. 
			
			For case (a),
			\begin{equation}\label{ll}
			\begin{aligned}
			\boldsymbol{\pi}_0'\boldsymbol{\Phi}\boldsymbol{\pi}_0&=\sigma_x^2\dfrac{m(1-d^2)-2d+2d^{m+1}}{m^2(1-d)^2},\\
			\boldsymbol{\pi}_0'\boldsymbol{\Phi}\boldsymbol{\pi}({\theta})&=\sigma_x^2\dfrac{(1+d)-\sum_{i=1}^{m}(d^i+d^{m+1-i})\pi_i}{m(1-d)},\\
			(\boldsymbol{\Phi}\boldsymbol{\pi}_0)_k&=\sigma_x^2\dfrac{1+d-d^{m-k+1}-d^{k}}{m(1-d)},
			\end{aligned}
			\end{equation}
			where $(\boldsymbol{\Phi}\boldsymbol{\pi}_0)_k$ is the $k$-th element of $\boldsymbol{\Phi}\boldsymbol{\pi}_0$.
			
			Based on (\ref{ll}), the second term of (\ref{case-ii}) follows
			\begingroup
			\allowdisplaybreaks
			\begin{align*}
			&\left|\Upsilon_r'(\boldsymbol{\pi}_0'\boldsymbol{\Phi}\boldsymbol{\pi}_0)^{-1}\boldsymbol{\pi}_0'\boldsymbol{\Phi}\boldsymbol{\pi}({\theta})\boldsymbol{\Phi}\boldsymbol{\pi}_0\right|=|(\boldsymbol{\pi}_0'\boldsymbol{\Phi}\boldsymbol{\pi}_0)^{-1}||\boldsymbol{\pi}_0'\boldsymbol{\Phi}\boldsymbol{\pi}({\theta})||\Upsilon_r'\boldsymbol{\Phi}\boldsymbol{\pi}_0|\\
			&=\dfrac{\sigma_x^2}{m(1-d^2)-2d+2d^{m+1}}\left|(1+d)-\sum_{i=1}^{m}(d^i+d^{m+1-i})\pi_i\right|\left|\sum_{k=1}^{m}({1+d-d^{m-k+1}-d^{k}})f_r(k)\right|\\
			&\leq\dfrac{\sigma_x^2}{|m(1-d^2)-2d|}\left(1+|d|+\left|\sum_{i=1}^{m}(d^i+d^{m+1-i})\pi_i\right|\right)\left(\sum_{k=1}^{m}({1+|d|+|d|^{m-k+1}+|d|^{k}})f_r(k)\right)\\
			&\leq\dfrac{\sigma_x^2}{m(1-d^2)-2|d|}\left(1+|d|+(|d|+|d|)\sum_{i=1}^{m}\pi_i\right)\left((1+|d|+|d|+|d|)\sum_{k=1}^{m}f_r(k)\right)\\
			&\leq\dfrac{\sigma_x^2(1+3|d|)^2}{m(1-d^2)-2|d|}=O(m^{-1}). \numberthis
			\end{align*}%
			\endgroup
			
			Hence, both the first term and the second term of (\ref{case-ii}) are $O(m^{-1})$ for two instruments. It follows that $E(\mathbf{z}_tu_t^A)=\left(O(m^{-1}),\ O(m^{-1})\right)'$.
			
			Now, consider case (b), the general case of the end-of-period sampling. We still assume that $\boldsymbol{\pi}_0=(\pi_{0,1},\cdots,\pi_{0,n},0,\cdots,0)'$ for any integer $n\in\left[0,m\right)$ independent of $m$ such that $\pi_{0,i}$ is positive constants independ of $m$ for all $i=1,\cdots,n$ and $\sum_{i=1}^{n}\pi_{0,i}=1$. Since we assume that only the first $n$ elements can be assigned with positive values which are no greater than 1, then the $k$-th element of $\boldsymbol{\pi}_0'\boldsymbol{\Phi}$ is
			\begin{equation}
			(\boldsymbol{\pi}_0'\boldsymbol{\Phi})_k=\left\{\begin{matrix}
			\sigma_x^2\left(\sum_{i=k}^{n}\pi_{0,i}d^{i-k}+\sum_{j=1}^{k-1}\pi_{0,j}d^j\right), & 1\leq k\leq n,\\
			\sigma_x^2d^{k-n}\sum_{p=1}^{n}d^{n-p}\pi_{0,p}, & n<k\leq m.
			\end{matrix}\right.
			\end{equation}
			
			Then, similar to the i.i.d. case, we can derive the followings for $r=1,2$.
			\begingroup
			\allowdisplaybreaks
			\begin{align*}
			\boldsymbol{\pi}_0'\boldsymbol{\Phi}\boldsymbol{\pi}_0&=\sigma_x^2\sum_{k=1}^{n}\left(\sum_{i=k}^{n}\pi_{0,i}d^{i-k}+\sum_{j=1}^{k-1}\pi_{0,j}d^j\right)\pi_{0,k}=\sigma_x^2D_0(d,n;\boldsymbol{\pi}_0),\\
			\boldsymbol{\pi}_0'\boldsymbol{\Phi}\boldsymbol{\pi}({\theta})&=\sigma_x^2\sum_{k=1}^{n}\left(\sum_{i=k}^{n}\pi_{0,i}d^{i-k}+\sum_{j=1}^{k-1}\pi_{0,j}d^j\right)\pi_{k}+\sigma_x^2\sum_{k=n+1}^{m}\left(d^{k-n}\sum_{p=1}^{n}d^{n-p}\pi_{0,p}\right)\pi_{k}\\
			&\leq\sigma_x^2\cdot \dfrac{2^{4\theta}}{S_{\pi}}\left(D_1(d,n;\boldsymbol{\pi}_0)+\left(\sum_{k=n+1}^{m}d^{k-n}\pi_{k}\cdot\sum_{p=1}^{n}d^{n-p}\pi_{0,p}\right)\right)\\
			&\leq\sigma_x^2\cdot \dfrac{2^{4\theta}}{S_{\pi}}\left(D_1(d,n;\boldsymbol{\pi}_0)+\dfrac{1-d^{m-n+1}}{1-d}D_2({d,n;\boldsymbol{\pi}_0})\right),\\
			\Upsilon_r'\boldsymbol{\Phi}\boldsymbol{\pi}_0&=\sigma_x^2\sum_{k=1}^{n}\left(\sum_{i=k}^{n}\pi_{0,i}d^{i-k}+\sum_{j=1}^{k-1}\pi_{0,j}d^j\right)f_r(k)+\sigma_x^2\sum_{k=n+1}^{m}\left(d^{k-n}\sum_{p=1}^{n}d^{n-p}\pi_{0,p}\right)f_r(k)\\
			&\leq\sigma_x^2\max_{1\leq k\leq m}f_r(k)\cdot\left(D_1(d,n;\boldsymbol{\pi}_0)+\dfrac{1-d^{m-n+1}}{1-d}D_2({d,n;\boldsymbol{\pi}_0})\right), \numberthis
			\end{align*}
			\endgroup
			where  
			$D_1(d,n;\boldsymbol{\pi}_0)=\sum_{k=1}^{n}\left(\sum_{i=k}^{n}\pi_{0,i}d^{i-k}+\sum_{j=1}^{k-1}\pi_{0,j}d^j\right)$ and $D_2(d,n;\boldsymbol{\pi}_0)=\sum_{p=1}^{n}d^{n-p}\pi_{0,p}$ relies on $d,\ n$ and $\boldsymbol{\pi}_0$.
			
			Therefore, we can derive that
			\begin{equation}
			\begin{aligned}
			&\left|\Upsilon_r'(\boldsymbol{\pi}_0'\boldsymbol{\Phi}\boldsymbol{\pi}_0)^{-1}\boldsymbol{\pi}_0'\boldsymbol{\Phi}\boldsymbol{\pi}({\theta})\boldsymbol{\Phi}\boldsymbol{\pi}_0\right|\\
			&\leq\dfrac{\sigma_x^2\cdot\max_{1\leq k\leq m}f_r(k)}{|D_0(d,n;\boldsymbol{\pi}_0)|}\cdot\dfrac{2^{4\theta}}{S_{\pi}}\cdot\left(D_1(d,n;\boldsymbol{\pi}_0)+\dfrac{1-d^{m-n+1}}{1-d}D_2({d,n;\boldsymbol{\pi}_0})\right)^2=O(m^{-1}).
			\end{aligned}
			\end{equation}
			
			Hence, both the first term and the second term of (\ref{case-ii}) are $O(m^{-1})$ for two instruments. It follows that $E(\mathbf{z}_tu_t^A)=\left(O(m^{-1}),\ O(m^{-1})\right)'$.
			
			Therefore, for either the i.i.d. or the AR(1) high-frequency regressor, $E(z_{r,t}u_t^A)=O(m^{-1})$ for $r=1,2$ can be satisfied with either the flat aggregation $\boldsymbol{\pi}_0=(1/m,\cdots,1/m)'$ or the general case of the end-of-period sampling $\boldsymbol{\pi}_0=(\pi_{0,1},\cdots,\pi_{0,n},0,\cdots,0)'$.
			
		\end{proof}

		\section{Full Simulation Results}\label{S4}

		All three methods perform similar sizes close to 0.05. By choosing our choice of instruments, larger powers are presented generally for large frequency ratios. However, our method does not perform larger powers for small frequency ratio, especially with small alternatives.   
		
		\begin{table}[H]
			\setlength\tabcolsep{1.6pt}
			\centering
			\scriptsize
			\caption[Caption for LOF]{Empirical Sizes and Powers for the Simulation Model: $T=125$, $m=4$}
\begin{tabular}{cccc|cccccccccccccccccccc}
				\hline
				d     & c     &       & 0.0   & 0.1   & 0.2   & 0.3   & 0.4   & 0.5   & 0.6   & 0.7   & 0.8   & 0.9   & 1.0   & 1.1   & 1.2   & 1.3   & 1.4   & 1.5   & 1.6   & 1.7   & 1.8   & 1.9   & 2.0 \\
				\hline
				\multirow{15}[10]{*}{-0.5} & \multirow{3}[2]{*}{-0.5} & Miller & 6.5 & 100 & 100   & 100   & 100   & 100   & 100   & 100   & 100   & 100   & 100   & 100   & 100   & 100   & 100   & 100   & 100   & 100   & 100   & 100   & 100 \\
				&       & AGK   & 5.7   & 100   & 100   & 100   & 100   & 100   & 100   & 100   & 100   & 100   & 100   & 100   & 100   & 100   & 100   & 100   & 100   & 100   & 100   & 100   & 100 \\
				&       & New   & \textbf{5.5}   & \textbf{0.1} & \textbf{0} & \textbf{0.4} & \textbf{22} & \textbf{83.1} & 99.2  & 100   & 100   & 100   & 100   & 100   & 100   & 100   & 100   & 100   & 100   & 100   & 100   & 100   & 100 \\
				\cline{2-24}
				& \multirow{3}[2]{*}{0.0} & Miller & 6.6 & 100 & 100  & 100   & 100   & 100   & 100   & 100   & 100   & 100   & 100   & 100   & 100   & 100   & 100   & 100   & 100   & 100   & 100   & 100   & 100 \\
				&       & AGK   & 6     & 100   & 100   & 100   & 100   & 100   & 100   & 100   & 100   & 100   & 100   & 100   & 100   & 100   & 100   & 100   & 100   & 100   & 100   & 100   & 100 \\
				&       & New   & \textbf{5.8}   & \textbf{0.4} & \textbf{0.1} & \textbf{0.8} & \textbf{22.6} & \textbf{80.1} & 99.2  & 100   & 100   & 100   & 100   & 100   & 100   & 100   & 100   & 100   & 100   & 100   & 100   & 100   & 100 \\
				\cline{2-24}
				& \multirow{3}[2]{*}{0.3} & Miller & 6.7   & 100 & 100    & 100  & 100   & 100   & 100   & 100   & 100   & 100   & 100   & 100   & 100   & 100   & 100   & 100   & 100   & 100   & 100   & 100   & 100 \\
\hline
			\end{tabular}
		\end{table}
		
		\begin{table}[H]
			\setlength\tabcolsep{0.8pt}
			\centering
			\scriptsize
			\caption[Caption for LOF]{Empirical Sizes and Powers for the Simulation Model: $T=125$, $m=150$}
			\begin{tabular}{cccc|cccccccccccccccccccc}
				\hline
				d     & c     &       & 0.0   & 0.1   & 0.2   & 0.3   & 0.4   & 0.5   & 0.6   & 0.7   & 0.8   & 0.9   & 1.0   & 1.1   & 1.2   & 1.3   & 1.4   & 1.5   & 1.6   & 1.7   & 1.8   & 1.9   & 2.0 \\
				\hline
				\multirow{15}[10]{*}{-0.5} & \multirow{3}[2]{*}{-0.5} & Miller & 6.1 & 100 & 100  & 100   & 100   & 100   & 100   & 100   & 100   & 100   & 100   & 100   & 100   & 100   & 100   & 100   & 100   & 100   & 100   & 100   & 100 \\
				&       & AGK   & \textbf{5.3}   & \textbf{31.4} & \textbf{45.3} & \textbf{52.4} & \textbf{58.7} & \textbf{62.1} & \textbf{66.2} & \textbf{69.5} & \textbf{73} & \textbf{76.1} & \textbf{78.2} & \textbf{79.8} & \textbf{81} & \textbf{82.4} & \textbf{83.5} & \textbf{83.9} & \textbf{84.8} & \textbf{85.3} & \textbf{85.7} & \textbf{86.4} & \textbf{86.9} \\
				&       & New   & 5.7   & 100   & 100   & 100   & 100   & 100   & 100   & 100   & 100   & 100   & 100   & 100   & 100   & 100   & 100   & 100   & 100   & 100   & 100   & 100   & 100 \\
				\cline{2-24}
				& \multirow{3}[2]{*}{0.0} & Miller & 6.1   & 100 & 100 & 100  & 100   & 100   & 100   & 100   & 100   & 100   & 100   & 100   & 100   & 100   & 100   & 100   & 100   & 100   & 100   & 100   & 100 \\
				&       & AGK   & \textbf{5.2} & \textbf{23.1} & \textbf{39.4} & \textbf{49.3} & \textbf{56.5} & \textbf{60.7} & \textbf{65} & \textbf{68.4} & \textbf{71.9} & \textbf{75.5} & \textbf{77.7} & \textbf{79.2} & \textbf{81.1} & \textbf{81.9} & \textbf{82.7} & \textbf{83.6} & \textbf{84.1} & \textbf{85} & \textbf{85.6} & \textbf{86.1} & \textbf{86.5} \\
				&       & New   & 5.6   & 100   & 100   & 100   & 100   & 100   & 100   & 100   & 100   & 100   & 100   & 100   & 100   & 100   & 100   & 100   & 100   & 100   & 100   & 100   & 100 \\
				\cline{2-24}
				& \multirow{3}[2]{*}{0.3} & Miller & 6.1 & 100 & 100 & 100 & 100  & 100    & 100  & 100  & 100  & 100  & 100  & 100  & 100   & 100 & 100   & 100  & 100  & 100 & 100  & 100   & 100 \\
\hline
			\end{tabular}
		\end{table}

		\begin{table}[H]
			\setlength\tabcolsep{0.8pt}
			\centering
			\scriptsize
			\caption[Caption for LOF]{Empirical Sizes and Powers for the Simulation Model: $T=125$, $m=365$}
			\begin{tabular}{cccc|cccccccccccccccccccc}
				\hline
				d     & c     &       & 0.0   & 0.1   & 0.2   & 0.3   & 0.4   & 0.5   & 0.6   & 0.7   & 0.8   & 0.9   & 1.0   & 1.1   & 1.2   & 1.3   & 1.4   & 1.5   & 1.6   & 1.7   & 1.8   & 1.9   & 2.0 \\
				\hline
				\multirow{15}[10]{*}{-0.5} & \multirow{3}[2]{*}{-0.5} & Miller & 7.1   & 100   & 100   & 100   & 100   & 100   & 100   & 100   & 100   & 100   & 100   & 100   & 100   & 100   & 100   & 100   & 100   & 100   & 100   & 100   & 100 \\
				&       & AGK   & 6.7   & \textbf{15.5} & \textbf{19.6} & \textbf{23.1} & \textbf{26.8} & \textbf{30} & \textbf{33.2} & \textbf{36.1} & \textbf{39.1} & \textbf{41.4} & \textbf{43} & \textbf{45.2} & \textbf{47.8} & \textbf{50.1} & \textbf{51.6} & \textbf{54} & \textbf{55.9} & \textbf{57.5} & \textbf{59.3} & \textbf{60.8} & \textbf{61.5} \\
				&       & New   & \textbf{4.7} & 100   & 100   & 100   & 100   & 100   & 100   & 100   & 100   & 100   & 100   & 100   & 100   & 100   & 100   & 100   & 100   & 100   & 100   & 100   & 100 \\
				\cline{2-24}
				& \multirow{3}[2]{*}{0.0} & Miller & 7.1 & 100   & 100   & 100   & 100   & 100   & 100   & 100   & 100   & 100   & 100   & 100   & 100   & 100   & 100   & 100   & 100   & 100   & 100   & 100   & 100 \\
				&       & AGK   & 6.5   & \textbf{13} & \textbf{17.7} & \textbf{21.1} & \textbf{25.1} & \textbf{28.7} & \textbf{32.6} & \textbf{35} & \textbf{38} & \textbf{40.7} & \textbf{42.5} & \textbf{44.6} & \textbf{46.8} & \textbf{49.5} & \textbf{51.5} & \textbf{53.4} & \textbf{55.6} & \textbf{57.1} & \textbf{58.8} & \textbf{60.1} & \textbf{61.1} \\
				&       & New   & \textbf{4.8}   & 100   & 100   & 100   & 100   & 100   & 100   & 100   & 100   & 100   & 100   & 100   & 100   & 100   & 100   & 100   & 100   & 100   & 100   & 100   & 100 \\
\hline
			\end{tabular}
		\end{table}

		\begin{table}[H]
			\setlength\tabcolsep{0.8pt}
			\centering
			\scriptsize
			\caption[Caption for LOF]{Empirical Sizes and Powers  for the Simulation Model: $T=512$, $m=4$}
			\begin{tabular}{cccc|cccccccccccccccccccc}
				\hline
				d     & c     &       & 0.0   & 0.1   & 0.2   & 0.3   & 0.4   & 0.5   & 0.6   & 0.7   & 0.8   & 0.9   & 1.0   & 1.1   & 1.2   & 1.3   & 1.4   & 1.5   & 1.6   & 1.7   & 1.8   & 1.9   & 2.0 \\
				\hline
				\multirow{15}[10]{*}{-0.5} & \multirow{3}[2]{*}{-0.5} & Miller & 5.5   & 100   & 100   & 100   & 100   & 100   & 100   & 100   & 100   & 100   & 100   & 100   & 100   & 100   & 100   & 100   & 100   & 100   & 100   & 100   & 100 \\
				&       & AGK   & 5.5   & 100   & 100   & 100   & 100   & 100   & 100   & 100   & 100   & 100   & 100   & 100   & 100   & 100   & 100   & 100   & 100   & 100   & 100   & 100   & 100 \\
				&       & New   & \textbf{5.4} & \textbf{6.2} & \textbf{4.4} & \textbf{1.3} & \textbf{0.1} & \textbf{0} & \textbf{0} & \textbf{0.1} & \textbf{0.2} & \textbf{0.5} & \textbf{5.2} & \textbf{27.1} & \textbf{71.5} & 96.1  & 100   & 100   & 100   & 100   & 100   & 100   & 100 \\
				\cline{2-24}
				& \multirow{3}[2]{*}{0.0} & Miller & 6.6   & 100   & 100   & 100   & 100   & 100   & 100   & 100   & 100   & 100   & 100   & 100   & 100   & 100   & 100   & 100   & 100   & 100   & 100   & 100   & 100 \\
				&       & AGK   & 5.2   & 99.9  & 100   & 100   & 100   & 100   & 100   & 100   & 100   & 100   & 100   & 100   & 100   & 100   & 100   & 100   & 100   & 100   & 100   & 100   & 100 \\
				&       & New   & \textbf{4.9} & \textbf{5.4} & \textbf{4.5} & \textbf{2.2} & \textbf{0.6} & \textbf{0.1} & \textbf{0.1} & \textbf{0.1} & \textbf{0.5} & \textbf{1.4} & \textbf{7.7} & \textbf{27} & \textbf{64.8} & 92.4  & 99.2  & 100   & 100   & 100   & 100   & 100   & 100 \\
				\cline{2-24}
				& \multirow{3}[2]{*}{0.3} & Miller & 6.6   & 99.9 & 100   & 100   & 100   & 100   & 100   & 100   & 100   & 100   & 100   & 100   & 100   & 100   & 100   & 100   & 100   & 100   & 100   & 100   & 100 \\
\hline
			\end{tabular}
		\end{table}

		\begin{table}[H]
			\setlength\tabcolsep{0.6pt}
			\centering
			\scriptsize
			\caption[Caption for LOF]{Empirical Sizes and Powers  for the Simulation Model: $T=512$, $m=150$}
			\begin{tabular}{cccc|cccccccccccccccccccc}
				\hline
				d     & c     &       & 0.0   & 0.1   & 0.2   & 0.3   & 0.4   & 0.5   & 0.6   & 0.7   & 0.8   & 0.9   & 1.0   & 1.1   & 1.2   & 1.3   & 1.4   & 1.5   & 1.6   & 1.7   & 1.8   & 1.9   & 2.0 \\
				\hline
				\multirow{15}[10]{*}{-0.5} & \multirow{3}[2]{*}{-0.5} & Miller & 6.4 & 100   & 100   & 100   & 100   & 100   & 100   & 100   & 100   & 100   & 100   & 100   & 100   & 100   & 100   & 100   & 100   & 100   & 100   & 100   & 100 \\
				&       & AGK   & \textbf{5.2} & \textbf{23.9} & \textbf{61.2} & \textbf{81.7} & 90.5  & 94.8  & 96.2  & 97.3  & 98    & 98.2  & 98.4  & 98.4  & 98.5  & 98.5  & 98.7  & 98.7  & 98.7  & 98.9  & 98.9  & 98.9  & 99 \\
				&       & New   & 5.8   & 100   & 100   & 100   & 100   & 100   & 100   & 100   & 100   & 100   & 100   & 100   & 100   & 100   & 100   & 100   & 100   & 100   & 100   & 100   & 100 \\
				\cline{2-24}
				& \multirow{3}[2]{*}{0.0} & Miller & 6.6   & 98.6   & 100   & 100   & 100   & 100   & 100   & 100   & 100   & 100   & 100   & 100   & 100   & 100   & 100   & 100   & 100   & 100   & 100   & 100   & 100 \\
				&       & AGK   & \textbf{5.2} & \textbf{13.2} & \textbf{37.5} & \textbf{62.5} & \textbf{78.2} & \textbf{86.6} & 92.2  & 94.9  & 96.1  & 97.1  & 97.8  & 98.1  & 98.3  & 98.5  & 98.5  & 98.6  & 98.6  & 98.7  & 98.9  & 98.9  & 98.9 \\
				&       & New   & 5.8   & 98.8  & 100   & 100   & 100   & 100   & 100   & 100   & 100   & 100   & 100   & 100   & 100   & 100   & 100   & 100   & 100   & 100   & 100   & 100   & 100 \\
				\cline{2-24}
				& \multirow{3}[2]{*}{0.3} & Miller & 6.8 & \textbf{85.2}   & 100   & 100   & 100   & 100   & 100   & 100   & 100   & 100   & 100   & 100   & 100   & 100   & 100   & 100   & 100   & 100   & 100   & 100   & 100 \\
\hline
			\end{tabular}
		\end{table}

		\begin{table}[H]
			\setlength\tabcolsep{0.6pt}
			\centering
			\scriptsize
			\caption[Caption for LOF]{Empirical Sizes and Powers for the Simulation Model: $T=512$, $m=365$}
			\begin{tabular}{cccc|cccccccccccccccccccc}
				\hline
				d     & c     &       & 0.0   & 0.1   & 0.2   & 0.3   & 0.4   & 0.5   & 0.6   & 0.7   & 0.8   & 0.9   & 1.0   & 1.1   & 1.2   & 1.3   & 1.4   & 1.5   & 1.6   & 1.7   & 1.8   & 1.9   & 2.0 \\
				\hline
				\multirow{15}[10]{*}{-0.5} & \multirow{3}[2]{*}{-0.5} & Miller & 5.7 & 100   & 100   & 100   & 100   & 100   & 100   & 100   & 100   & 100   & 100   & 100   & 100   & 100   & 100   & 100   & 100   & 100   & 100   & 100   & 100 \\
				&       & AGK   & 5.4   & \textbf{11} & \textbf{25.1} &\textbf{ 40.3} & \textbf{50.5} & \textbf{57.1} & \textbf{62.4} & \textbf{66.5} & \textbf{68.5} & \textbf{70.7} & \textbf{73.3} & \textbf{74.7} & \textbf{76.9} & \textbf{78.2} & \textbf{79.5} & \textbf{80.3} & \textbf{81.5} & \textbf{82.2} & \textbf{83} & \textbf{83.8} & \textbf{84.7} \\
				&       & New   & \textbf{5.2} & 100   & 100   & 100   & 100   & 100   & 100   & 100   & 100   & 100   & 100   & 100   & 100   & 100   & 100   & 100   & 100   & 100   & 100   & 100   & 100 \\
				\cline{2-24}
				& \multirow{3}[2]{*}{0.0} & Miller & 5.9 & 98.9   & 100   & 100   & 100   & 100   & 100   & 100   & 100   & 100   & 100   & 100   & 100   & 100   & 100   & 100   & 100   & 100   & 100   & 100   & 100 \\
				&       & AGK   & 5.3   & \textbf{8.1} & \textbf{15.3} & \textbf{25.7} & \textbf{36.6} & \textbf{45.2} & \textbf{52.1} & \textbf{57.7} & \textbf{62.3} & \textbf{65.5} & \textbf{68.6} & \textbf{70.7 }& \textbf{73.1} & \textbf{75.1} & \textbf{77} & \textbf{78.4} & \textbf{79.9} & \textbf{81} & \textbf{81.9} & \textbf{82.7} & \textbf{83.3} \\
				&       & New   & \textbf{5.1}   & 99.3  & 100   & 100   & 100   & 100   & 100   & 100   & 100   & 100   & 100   & 100   & 100   & 100   & 100   & 100   & 100   & 100   & 100   & 100   & 100 \\
\hline
			\end{tabular}
		\end{table}

	\end{appendices}

	\newpage
	
	\medskip

	\bibliography{References}

\begin{thebibliography}{10}
\providecommand{\natexlab}[1]{#1}
\providecommand{\url}[1]{\texttt{#1}}
\expandafter\ifx\csname urlstyle\endcsname\relax
  \providecommand{\doi}[1]{doi: #1}\else
  \providecommand{\doi}{doi: \begingroup \urlstyle{rm}\Url}\fi

\bibitem[Andreou et~al.(2010)Andreou, Ghysels, and Kourtellos]{andreou}
Elena Andreou, Eric Ghysels, and Andros Kourtellos.
\newblock Regression models with mixed sampling frequencies.
\newblock \emph{Journal of Econometrics}, 158:\penalty0 246--261, 2010.

\bibitem[Andrews(1991)]{Andrews:1991}
Donald W~K Andrews.
\newblock {Heteroskedasticity and autocorrelation consistent covariance matrix
  estimation}.
\newblock \emph{Econometrica}, 59\penalty0 (3):\penalty0 817--858, 1991.

\bibitem[Foroni and Marcellino(2013)]{foroni2013survey}
Claudia Foroni and Massimiliano~Giuseppe Marcellino.
\newblock A survey of econometric methods for mixed-frequency data.
\newblock 2013.

\bibitem[Ghysels et~al.(2004)Ghysels, Santa-Clara, and Valkanov]{2004}
Eric Ghysels, Pedro Santa-Clara, and Rossen Valkanov.
\newblock The {MIDAS} touch: Mixed data sampling regression models.
\newblock \emph{Finance}, 2004.

\bibitem[Greene(2012)]{Greene}
William~H. Greene.
\newblock \emph{Econometric Analysis, 7th Edition}.
\newblock Prentice Hall, Upper Saddle River, 2012.

\bibitem[Groenvik and Rho(2018)]{Groenvik:Rho:2017}
Henriette Groenvik and Yeonwoo Rho.
\newblock A self-normalizing approach to the specification test of
  mixed-frequency models.
\newblock \emph{Communications in Statistics - Theory and Methods}, 47\penalty0
  (8):\penalty0 1913--1922, 2018.

\bibitem[Lee(2010)]{Lee:2010}
{Myoung Jae} Lee.
\newblock \emph{Micro-econometrics: Methods of moments and limited dependent
  variables (Second Edition)}.
\newblock Springer New York, 12 2010.

\bibitem[Miller(2018)]{Miller2016}
Isaac~J. Miller.
\newblock Simple robust tests for the specification of high-frequency
  predictors of a low-frequency series.
\newblock \emph{Econometrics and Statistics}, 5:\penalty0 45 -- 66, 2018.
\newblock ISSN 2452-3062.

\bibitem[Newey and West(1987)]{Newey:West:1987}
Whitney Newey and Kenneth~D. West.
\newblock A simple, positive semi-definite, heteroskedasticity and
  autocorrelation consistent covariance matrix.
\newblock \emph{Econometrica}, 55\penalty0 (3):\penalty0 703--708, 1987.

\bibitem[Ruud(2000)]{Ruud:2000}
Paul Ruud.
\newblock \emph{An Introduction to Classical Econometric Theory}.
\newblock Oxford University Press, 2000.

\end{thebibliography}

\end{document}